\documentclass[11pt]{article}
\usepackage{amsthm}
\usepackage{amsmath}
\usepackage{amssymb}
\usepackage{relsize}
\usepackage{nicefrac} 
\usepackage[margin=1.5in]{geometry}

\newtheorem{theorem}{Theorem}[section]
\newtheorem{lemma}[theorem]{Lemma}
\newtheorem{corollary}[theorem]{Corollary}
\newtheorem{proposition}[theorem]{Proposition}

\newtheorem{definition}[theorem]{Definition}

\renewcommand{\P}{\mathrm{P}}
\newcommand{\NP}{\mathrm{NP}}
\newcommand{\coNP}{\mathrm{coNP}}
\newcommand{\FP}{\mathrm{FP}}
\newcommand{\PV}{\mathrm{PV}}
\newcommand{\TFNP}{\mathrm{TFNP}}

\newcommand{\NN}{\mathbb{N}}

\newcommand{\PLS}{\mathrm{PLS}}
\newcommand{\CPLS}{\mathrm{CPLS}}
\def\GI{\mathrm{GI}}
\def\FP{\mathrm{FP}}

\def\PPA{\mathrm{PPA}}
\def\PPAD{\mathrm{PPAD}}
\def\PPADS{\mathrm{PPADS}}
\def\PPP{\mathrm{PPP}}
\def\PLC{\mathrm{PLC}}

\def\APPROX{\mathrm{APPROX}}
\def\TFS{\mathrm{TF}\Sigma^p_2}
\def\RWPP{\textsc{RWPP}}
\def\APC{\mathrm{APC}}

\newcommand{\hf}{\nicefrac{1}{2}}

\def\Search{\mathrm{Search}}
\def\dt{^\mathrm{dt}}
\def\Out{\mathrm{Out}}
\def\poly{\mathrm{poly}}
\def\LK{\mathrm{LK}}
\def\PK{\mathrm{PK}}

\def\CNF{\mathrm{CNF}}

\def\Gqbf{\Gamma_{\mathrm{TQBF}}}
\def\Gpty{\Gamma_\oplus}
\def\Gps{\Gamma_{\mathrm{PSPACE}}}
\def\Gsh{\Gamma_{\#}}

\newcommand{\loc}{\textsc{Local-}}
\newcommand{\chk}{\textsc{Checkable~}}

\newcommand{\Reach}{\mathrm{Reach}}

\newcommand{\restrict} {\upharpoonright}
\newcommand{\pty}{\text{\raisebox{0.5pt}{\Large{$\oplus$}}}}

\newcommand{\inp}{\mathrm{inp}}
\newcommand{\eval}{\mathrm{eval}}
\newcommand{\FF}{\mathbb{F}}

\newcommand{\FCon}{\textsc{FCon}}

\newcommand{\bigdoublevee}{%
  \mathop{
    \mathchoice{\bigvee\mkern-15mu\bigvee}
               {\bigvee\mkern-12.5mu\bigvee}
               {\bigvee\mkern-12.5mu\bigvee}
               {\bigvee\mkern-11mu\bigvee}
    }
}

\newcommand{\oneref}{1\textsc{-Ref}}

\let\phi\varphi
\let\epsilon\varepsilon

\newcommand{\ignore}[1]{}

\begin{document}
\title{How to fit large complexity classes into TFNP}

\author{Neil Thapen\footnote{
Institute of Mathematics, Czech Academy of Sciences,  thapen@math.cas.cz. 
Supported by the Czech Academy of Sciences (RVO 67985840) and GA\v{C}R grant 23-04825S.}}

\maketitle

\begin{abstract}
Subclasses of TFNP (total functional NP) are usually defined by specifying a complete problem,
which is necessarily in TFNP, and including all problems many-one reducible to it.
We study two notions of how a TFNP problem can be reducible to an object,
such as a complexity class, outside TFNP. This gives rise to subclasses of TFNP
which capture some properties of that outside object.
We show that well-known subclasses can arise in this way,
for example PPA from reducibility to $\pty \P$ 
and PLS from reducibility to $\P^\NP$.

We study subclasses arising from PSPACE and the polynomial
hierarchy, and show that they are characterized by the propositional
proof systems Frege and constant-depth Frege, extending
the known pairings between natural TFNP subclasses and proof systems.

We study approximate counting from this point of view, and look
for a subclass of TFNP that gives a natural home to combinatorial principles
such as {\sc Ramsey} which can be proved using approximate counting.
We relate this to the recently-studied {\sc Long choice} and 
{\sc Short choice} problems.
\end{abstract}

\section{Introduction}

An \emph{$\NP$ search problem} is specified by a polynomial-time relation
$R(x,y)$ and a polynomial $p$ such that for every string~$x$
there is a string $y$ with $|y| \le p(|x|)$ such that~$R(x,y)$. 
We call~$x$ an \emph{input} or \emph{instance}, and~$y$ a \emph{solution}. The search
problem is: given~$x$,  find such a~$y$.
The class of such problems is called TFNP,
for Total Functional NP
~\cite{johnson1988easy, megiddo1991total}.
We will often not write the length bound on~$y$,
 which may be assumed to be implicit in 
$R(x,y)$, and will use just $R(x,y)$ or $R$
as the name of the problem. 

Let $R(x,y)$ and $Q(u,v)$ be two such problems. Then
$Q$ is \emph{many-one reducible}, or simply \emph{reducible}, to $R$ if there 
are polynomial-time functions~$f$ and~$g$ such that for all $u, y$ we have
$R(f(u), y) \rightarrow Q(u, g(u,y))$. In other words, 
in polynomial time we can convert 
an instance~$u$ of $Q$  into an instance~$f(u)$
of $R$, and then convert any solution back into a solution of $Q$.  
We write this as $Q \le R$, and it gives rise to a structure on TFNP by
which we can classify search problems into 
natural subclasses, closed under reductions
 and usually defined by complete problems.
 
 $\TFNP$ has many connections with logic 
 (not least of which is the characterization of black-box $\TFNP$ classes
 in terms of propositional proof systems, 
 from which many separation results come, 
  e.g.~\cite{buresh2004relativized, goos2022separations}).
  One common observation is that many natural $\TFNP$
  subclasses  collect together all the problems in which the existence
  of a solution is guaranteed by some particular existence theorem
  or combinatorial lemma; the complete problem for the class
  is a particular case of the lemma~\cite{megiddo1991total, papadimitriou1994complexity}.
For example in~\cite{beame1995relative}, the class $\PPA$ 
``is based on the lemma that every graph has an even number of odd-degree nodes''.
{\sc Leaf} (defined below) is a complete problem for this class, 
and the totality of {\sc Leaf}, that is, the statement that every instance of {\sc Leaf} has a solution,
is a case of this lemma.

This line of thought leads to a general recipe for defining $\TFNP$ classes.
Working in first-order logic, choose some language which is rich enough 
to express suitable statements about polynomial time machines, and let $T$
be \emph{any} theory in this language. Then we can define a subclass of $\TFNP$
as the set of  problems provably total in~$T$,
that is,  the set $\{ R \in \TFNP : T \vdash \forall x \exists y R(x,y) \}$.\footnote{
This set may not be closed under many-one reductions.
To get a well-behaved class, we can consider instead the closure of this set under 
reductions, or equivalently strengthen
$T$ by adding every true universal $L_\PV$ sentence.
See e.g.~\cite[Lemma~4]{kol2022approx} or Lemma~\ref{lem:Gamma_over_PV} below.}
This approach works well with theories of bounded arithmetic;
see many citations throughout this paper, in particular~\cite{buss1994application}
which characterizes $\PLS$ as arising in this way from the theory~$T^1_2$.
The goal of this kind of work is often to understand and compare the strength of theories,
motivated partly by analogous work in classical proof theory
that measures theories by their provably recursive functions.

As a way of constructing interesting $\TFNP$ classes, an advantage of this approach,
compared to defining a class by specifying a complete problem, 
is that $T$ does not have to be associated with a $\TFNP$ problem but can be any theory we like,
and as complicated as we like. 
It could be all of ZFC, or express that ``every $\P^\NP$ machine has a computation
on every input'' or that ``a well-behaved parity function for circuits exists'', and we will get a 
corresponding class, capturing something about the theory. 
The disadvantage is that the class may have no natural complete problem,
and in the worst case we may have to work with formal proofs to show membership;
and we may get no insight into how solving problems in the class is computationally
related to $\P^\NP$, or parity, or whatever object the theory describes.

The goal of this paper is to present ways of defining $\TFNP$ problems
and classes that keep the advantage above, that we can construct classes
inside $\TFNP$ which capture some of the properties of much more complicated
objects or ideas, but that avoid the disadvantages; we will use constructions
that have clear computational meanings, and our classes will have natural
complete problems. Sometimes our constructions will parallel work done in bounded
arithmetic, but we  avoid arguments via logic when we can.
We develop two related constructions. 

\subsection{TFNP subclasses from complexity classes}

The first construction
starts with a complexity class, such as $\pty \P$,
and  (almost) canonically defines a $\TFNP$ subclass from it, 
consisting of problems that can be solved by a polynomial-time
interaction with an untrustworthy Adversary who claims to be able
to  answer queries to the class. 
The ``almost'' part is that
we need to specify a particular axiomatization of the complexity class;
the Adversary is allowed to give incorrect answers, as long as they do
not immediately violate any of our  axioms. 
This gives well-behaved classes with natural complete problems,
and our main results here are:
\begin{itemize}
\item
The class arising from $\pty \P$
is the standard class $\PPA$ (Theorem~\ref{the:Gpty_PPA}).
This does not depend on a particular
choice of axioms for $\pty \P$ (Theorem~\ref{the:robust_PPA}).
\item
The class arising from $\NP$, with natural axioms,
is the standard class $\PLS$
(Proposition~\ref{pro:QBF1_from_PNP}).
\item
Two different axiomatizations of PSPACE,
in terms of PSPACE computations and of true QBFs,
give rise to the same class (Theorem~\ref{the:pspace_equivalence}). 
This contains the standard classes $\PLS$, $\PPA$
and $\PPP$ by natural reductions
(Proposition~\ref{pro:QBF1}, Theorems~\ref{the:PPA_Gsh}
and~\ref{the:PPT_Gsh})
and  is characterized by the well-studied
Frege proof system (Proposition~\ref{pro:proof_characterizations}).
\item
The classes arising from the polynomial hierarchy
 coincide with previously-studied classes 
defined by natural combinatorial  problems
(Proposition~\ref{pro:GI_loc_k}) and 
are characterized by the well-studied constant-depth
Frege hierarchy of proof systems (Proposition~\ref{pro:proof_characterizations}).
\end{itemize}

Here a proof system ``characterizes'' a TFNP subclass if the subclass
contains precisely the problems which have small proofs of totality in the system.
It was recently shown in~\cite{goos2022separations} that every standard TFNP
class (except PPP) has a corresponding proof system; we 
exhibit classes for some relatively strong proof systems which did not
appear in this picture before. 

The Frege system in particular is known to be strong, in the sense that 
many natural combinatorial principles have quasipolynomial-sized proofs in it,
and lower bounds for it seem to be 
well out of reach (see e.g. the discussion in~\cite{aisenberg2018short}). 
For this reason, we like to think of the PSPACE class, corresponding to Frege, as sitting like a hat 
on top of the standard TFNP classes; any natural combinatorial TFNP problem 
is very likely to lie inside it.\footnote{Exceptions are the {\it local improvement}
search problems
corresponding to the theory $V^1_2$ from~\cite{kol2011so, beckmann2014improved},
which are plausibly characterized by extended Frege.} At the same time it
does not seem excessively strong, and is plausibly weaker than
the problem {\sc Wrong proof} proposed for this purpose 
in~\cite{goldberg2018towards}, which is a consistency principle for a proof
system above extended Frege.

\subsection{TFNP subclasses from TF$\boldsymbol{\Sigma}^p_2$  problems}

The second construction starts with a search problem in $\TFS$. This is a complexity
class that sits one quantifier level above $\TFNP$, in that a $\TFS$
problem is defined like a $\TFNP$ problem, except that
the relation $R(x,y)$ is $\coNP$ rather than polynomial-time.
In this sense, a string $y$ is a solution to an instance~$x$
of $R$ if and only if there is no \emph{counterexample}
to the $\coNP$ statement $R(x,y)$.

We say that a $\TFS$ problem $Q(u,v)$ is \emph{counterexample reducible}
to a $\TFS$ problem~$R(x,y)$ if there are polynomial-time functions $f$
and $g$ resembling a normal many-one reduction of $Q$ to $R$, 
and a third polynomial-time function~$h$ which converts
any counterexample to $Q(u,v)$ into a counterexample to $R(x,y)$.
Since $\TFNP$ problems are in particular $\TFS$ problems,
any $\TFS$ problem  gives rise to a subclass of $\TFNP$
consisting of the problems
which are counterexample-reducible to it.
Again this gives well-behaved classes with natural complete problems,
and our main results are:
\begin{itemize}
\item
The class arising from the $\TFS$ problem {\sc $\P^\NP$ computation} is precisely PLS
(Theorem~\ref{the:TFS_equivalences}).
\item
The class arising from the $\TFS$ problem {\sc Empty} is precisely PPADS.
Thus, relative to some oracle, there is no counterexample reduction in either
direction between {\sc Empty} and {\sc $\P^\NP$ computation}
(Theorem~\ref{the:TFS_equivalences}, Corollary~\ref{cor:empty_separation}).
\end{itemize}
 $\TFS$ is an increasingly popular topic~\cite{
kleinberg2021total, ren2022range, korten2022hardest, pasarkar2023extremal},
and is often studied under a coarser notion of reducibility, under which for
example all TFNP problems are reducible to {\sc Empty}~\cite{kleinberg2021total}.
Our results show that counterexample reducibility is a fruitful definition,
and we believe that it shows the richness of this class,
and in particular shows that many TFNP classes can be most naturally seen
as arising from $\TFS$ problems under counterexample reducibility.


\subsection{{\sc Ramsey}, PLC and approximate counting}

As an application of the machinery of counterexample reductions,
we give a simplified definition of the approximate-counting TFNP class APPROX
from~\cite{kol2022approx}.
We~show a direct reduction of {\sc Ramsey} to it, and compare it with weakened
versions of the counting problems {\sc Long choice} and {\sc Short choice}
from~\cite{pasarkar2023extremal}. Here we describe some of the background to this work.

It is a long-running research theme in logic to study which kind of counting arguments
can be carried out in
arithmetic without exponentiation, that is, in bounded arithmetic. In terms of the concepts
introduced in this paper, this is like asking how much counting we can do in the 
TFNP subclasses arising from the polynomial hierarchy.
Landmark results are that the weak pigeonhole principle is provable, and useful~\cite{pww:primes};
that the finite Ramsey theorem is provable~\cite{pudlak1990ramsey};
that precise counting is impossible, since the full pigeonhole principle is not 
provable~\cite{ajtai1994complexity, pitassi1993exponential, krajivcek1995exponential};
but that we can carry out any reasonable argument for which approximate
counting suffices, like the standard proofs of the Ramsey theorem and the tournament 
principle~\cite{jerabek:apc2}.
Here ``reasonable argument'' means roughly: we can do inductions with polynomially-many
steps, where the inductive hypothesis contains statements like ``set $X \subseteq [2^n]$
has size approximately~$s$'' where $X$ is polynomial-time and $s$ is correct up to an 
error~$s \epsilon$. It is shown in~\cite{jerabek:apc2} that such arguments can be carried out in a 
subtheory~$\APC_2$ of bounded arithmetic.
$\APC_2$ is strong enough to formalize many results of complexity theory,
and it and related theories are popular now in research in metacomplexity,
which seeks to formally understand why progress in complexity theory is hard
(see e.g.~\cite{muller2020feasibly, li2023unprovability}).

The paper~\cite{kol2022approx} introduced a TFNP class $\APPROX$
corresponding to $\APC_2$, as a tool to show an unprovability result.
$\APPROX$ suffered there from having a complicated definition, 
and in particular for not having any clear complete problem.
In Section~\ref{sec:approx} we recap results about APPROX 
from~\cite{jerabek:apc2, kol2022approx}. Our contribution is a new,
simpler definition of APPROX, using reducibility
to a $\TFS$ form of weak pigeonhole principle~(Proposition~\ref{pro:APPROX_PLS}).
This allows us to show in particular that {\sc Ramsey} is in APPROX by a direct,
``combinatorial'' reduction, rather than by an indirect argument through logic.

We also study the TFNP problem 
{\sc Long choice}, which was introduced recently in~\cite{pasarkar2023extremal}
with the goal of capturing the kind of counting that proves the Ramsey theorem.
Our contention is that, if we are interested in {\sc Ramsey} specifically,
 this problem and the corresponding class $\PLC$
are {too strong}. 
This is because in the framework of Inspector-Adversary games
we have
 a recipe for making TFNP classes that are as strong as you like;
any reasonable counting argument can be carried out in
$\loc\Gps$, and plausibly in $\loc\Gsh$ (see Section~\ref{sec:pspace}).
So the interesting question becomes, not just to find a natural class that is strong enough
to contain {\sc Ramsey} and similar problems,
but to find one that is also not too strong. In particular
it is arguably desirable that it does not contain the classes $\PPP$, or even $\PPAD$, 
associated with precise counting, since this is not necessary for {\sc Ramsey}.\footnote{
On the other hand it seems that it ought to contain
$\PLS$, since this is reducible to approximate counting ---
see the proof of Proposition~\ref{pro:approx_list}.
It should certainly contain {\sc Weak pigeon},
and in particular a (non-uniform) reduction of {\sc Weak pigeon} to {\sc Ramsey}
was shown in~\cite{krajicek2005structured}.}
 Here APPROX is a candidate, sitting, apparently strictly,
 inside $\loc\Gqbf^3$ (Propositions~\ref{pro:approx_in_Gqbf3}
 and~\ref{pro:CPLS_not_in_approx}) and so not containing $\PPAD$
 (Corollary~\ref{cor:hierarchy_separate}).
However, we do not have any evidence that a strictly-weaker class
would not be enough.

In Section~\ref{sec:long_choice} we define {\sc Long choice}
and introduce a variant {\sc Weak long choice}, with a similar
relationship to it as {\sc Weak pigeon} has to {\sc Pigeon}. 
We observe that this does not trivialize the problem, in that {\sc Ramsey}
and {\sc Weak pigeon} are still reducible to {\sc Weak long choice};
so in some sense the motivating idea behind {\sc Long choice}
is still present in {\sc Weak long choice}.
We show that {\sc Weak long choice} is in APPROX
(Proposition~\ref{pro:wlc_to_rwpp2}).

In Section~\ref{sec:short_choice} we show something similar
for the the $\TFS$ problem {\sc Short choice}, which was introduced
in~\cite{pasarkar2023extremal} as a kind of dual of {\sc Long choice}.
We define a $\TFS$ problem {\sc Weak short choice} and show
that it is non-trivial. We use counterexample reducibility to compare it to APPROX,
showing that every TFNP problem counterexample reducible to it is in APPROX
(Proposition~\ref{pro:weak_short_approx}).

\subsection{Open problems}

\noindent\textbf{1.} {Separate the polynomial hierarchy in TFNP.}
Precisely, show a relativized separation between $\loc\Gqbf^k$ and $\loc\Gqbf^{k+1}$
for any $k \ge 2$. 
This is the major open problem in this area.
It has a long history in proof complexity,
where it appears in the forms:
find a family of polylogarithmic width CNFs super-quasipolynomially separating
depth-$k$ from depth-$(k{+}1)$ Frege; or 
show that the relativized theory $T^{k+1}_2$ proves strictly more $\forall \Sigma^b_1$
sentences than $T^k_2$. 
As far as we know the current frontiers of what is known are a
separation of a ``linear'' form of the hierarchy~\cite{impkra_02}, 
a lower bound for the weak pigeonhole principle CNF
in a restricted version of Res(log)~\cite{razborov2015pseudorandom},
and that APPROX does not contain CPLS, or equivalently~$\loc\Gqbf^2$~\cite{kol2022approx}.

\medskip

\noindent\textbf{2.} {Find natural axioms $\Gamma$ capturing PPAD and PPADS 
in the style of $\Gpty$ and PPA  in Theorem~\ref{the:Gpty_PPA}.}
It is potentially easy to come up with ``unnatural'' axioms $\Gamma$ which work, for
example for PPAD taking an axiom expressing ``$\alpha(x)$ is a solution
for instance $x$ of {\sc Onto pigeon}''. We have in mind axioms that do
not necessarily describe any fixed object (such as $\pty \P$) but 
that set up reasonable rules about how the Adversary answers questions about cardinality.
It would also be interesting to set up some similar game where
we ask the Adversary the size of a set, and they are allowed some leeway
to give imprecise answers, and to study whether APPROX can be 
characterized in this way.

\medskip

\noindent\textbf{3.} Prove more general results of the form:
any two natural axiomatizations of a complexity class $\mathcal{C}$
give rise to the same TFNP subclass. This asks for a more general version
of Theorem~\ref{the:robust_PPA},
which shows that any recursive definition of $\pty \P$ with a constant
number of queries gives rise to PPA, or of Theorem~\ref{the:pspace_equivalence},
which shows that two very different characterizations of PSPACE
give rise to the same subclass.
This may be
connected to a phenomenon in proof complexity that, up to equivalence, a proof system is 
often determined by the expressiveness of lines in the proof, and the precise choice of rules does not
matter; for example Reckhow showed this for Frege systems in~\cite{reckhow1975lengths}.

\medskip

Let us also note one piece of work left undone in this paper.
We introduce a TFNP problem $\loc\Gps$ and show that it 
is connected to PSPACE and the Frege proof system via the theory~$U^1_2$.
Similar machinery should work for \mbox{EXPTIME}. That is, we expect
that a similar TFNP problem $\loc\Gamma_\mathrm{EXPTIME}$
is connected to EXPTIME and the extended Frege proof
system via the theory~$V^1_2$, using work 
from~\cite{kol2011so, beckmann2014improved, beckmann2017np}.

\subsection{Outline}

In the next section we give some basic definitions, of TFNP classes and of the language we will use
for writing axioms.

In Section 3 we define Inspector-Adversary games, our general model for creating a TFNP
subclass from an axiomatization of a complexity class. We show that this subclass has
a natural complete problem, and discuss  related constructions in logic.

In Section 4 we show that the subclass arising from $\pty \P$ is PPAD,
and that up to a point this is independent of which axioms we pick for $\pty \P$.

In Section 5 we study the subclass arising from PSPACE, axiomatized either by
the properties of QBFs or of PSPACE computations. We show that it contains the 
standard TFNP classes, and define subclasses arising from the polynomial
hierarchy by restricting the complexity of QBFs.

In Section 6 we define $\TFS$ and counterexample reducibility, and show that 
some standard TFNP subclasses arise in this way from simple $\TFS$ problems.

In Section 7 we study approximate counting and the Ramsey theorem in TFNP.
We give a simplified definition of the class APPROX
from~\cite{kol2022approx}, using the machinery of counterexample reductions.
We give a direct reduction of {\sc Ramsey} to it, and compare it with weakened
versions of the counting problems {\sc Long choice} and {\sc Short choice}
from~\cite{pasarkar2023extremal}.

In Section 8 we introduce decision-tree TFNP, and show that in this model
the  problems $\GI_k$ from~\cite{skelley2011provably} are characterized by constant-depth
Frege, and the problem $\FCon$ from~\cite{beckmann2017np} is characterized
by Frege.

In Section 9 we use results about bounded arithmetic from~\cite{skelley2011provably,
kol2011so, beckmann2017np} to show that the problems $\GI_k$ are equivalent to our
 problems arising from the polynomial hiearchy, and
that $\FCon$ is equivalent to our problem arising from PSPACE.

\medskip

The results of the broadest interest in this paper are probably the
new characterizations of TFNP classes by Frege and constant-depth Frege,
and we should note that these are more-or-less  implicit in the setting of bounded arithmetic
in~\cite{skelley2011provably} and~\cite{beckmann2017np}.
However those papers are concerned with slightly different questions
 and it takes a non-trivial amount of work to get them into the form we need.
Many other ideas come from taking things from proof complexity or bounded arithmetic and recontextualizing
them as computational constructions about TFNP or $\TFS$, 
in a way that we hope clarifies them and leads to new connections. 
Particularly new things are:
the notion of defining a TFNP class in a canonical way from a complexity class
and the work in this direction in Section~\ref{sec:robustness}; the 
TFNP problem $\loc\Gsh$ and its class;  the notion of counterexample reducibility; the 
reductions to APPROX in Section~\ref{sec:approx_counting}.

\paragraph{Acknowledgments.}
I am grateful to Sam Buss, Pavel Hub\'{a}\v{c}ek, Erfan Khaniki,
Leszek Ko\l{}odziejcyk, Jan Kraj\'{i}\v{c}ek and Pavel Pudl\'{a}k for discussions about this topic.

\section{Preliminaries} \label{sec:preliminaries}

We write $[a]$ for the interval $\{0, \dots, a{-}1\}$. Where
convenient we will identify numbers with their expressions as 
binary strings, and in particular 
identify $[2^n]$ with the set of binary strings of length~$n$.

When we discuss Boolean circuits we take them to have unlimited fan-in, with non-input
gates from $\bigwedge, \bigvee, \neg, 0, 1$.
For a circuit $C$ we write $\inp(C)$ for its set of input gates, 
for which we will use names
$x$, $y$, $x_1$, etc. For a gate $x \in \inp(C)$ and a bit $b \in \{0,1\}$
we write $C_{\restrict x=b}$ for the result of hard-wiring the value~$b$
into the gate~$x$, which is then no longer an input gate. 
We do not simplify the circuit after this operation. Note that the notation $\restrict \! \! x {=} b$ means that the gate $x$ is assigned value~$b$,
not that the variable $x$ (ranging over gates), is assigned  value~$b$.
If a circuit $C$ has no input gates, $\eval(C)$ evaluates it.

\paragraph{TFNP.}
%
We defined $\TFNP$ problems and many-one reductions at the start of the introduction.
Problems $Q$ and $R$ are equivalent, $Q \equiv R$, if there are reductions in both directions.
We will also write $\mathcal{S} \le R$ for a class $\mathcal{S}$ of $\TFNP$
problems, to mean that every problem in~$\mathcal{S}$ is reducible to $R$.
We say \emph{$R$ is complete for $\mathcal{S}$} if $\mathcal{S} \le R$ and $R \in \mathcal{S}$.
We list some standard $\TFNP$ problems and classes we will need.

\begin{definition} 
[\cite{johnson1988easy, papadimitriou1994complexity, beame1995relative, 
buresh2004relativized, buss2012propositional}] \label{def:standard_TFNP}
~
\begin{enumerate}
\item
\textsc{$\P$ computation} is the $\TFNP$ problem:
given a description $M$ of a Turing machine and a unary time
bound $1^t$, find a length-$t$ computation of $M$.

$\FP$ is the set of $\TFNP$ problems reducible to {\sc $\P$ computation}.
\item
\textsc{Pigeon} is the $\TFNP$ problem:
given a circuit defining a function $f$ from a set~$[2^n]$ of pigeons
to a set~$[2^n-1]$ of holes, find distinct pigeons $x, x'$
such that $f(x)=f(x')$.

$\PPP$ is the set of $\TFNP$ problems reducible to {\sc Pigeon}.
\item
{\sc Left pigeon} is the $\TFNP$ problem:
given circuits for functions $f$ from a set~$[2^n]$
of pigeons to a set~$[2^n-1]$ of holes, and $g$ mapping
in the other direction, from holes to pigeons,
find a pigeon $x$ such that $g(f(x)) \neq x$. 

$\PPADS$ is the set of $\TFNP$ problems reducible to {\sc Left pigeon}.

\item
{\sc Onto pigeon} is the $\TFNP$ problem:
given circuits for functions $f$ and $g$ as in {\sc Left pigeon},
find either a pigeon $x$ such that $g(f(x)) \neq x$
or a hole~$y$ such that $f(g(y)) \neq y$.

$\PPAD$ is the set of $\TFNP$ problems reducible to {\sc Onto pigeon}.

\item
\textsc{Leaf} is the $\TFNP$ problem:
given a circuit defining a graph on $[2^n]$ with nodes of degree at most~2,
in which node $0$ has degree~1, find another node of degree~1.

$\PPA$ is the set of $\TFNP$ problems reducible to {\sc Leaf}.
\item
\textsc{Lonely} is the $\TFNP$ problem:
given a circuit describing a partial matching on $[2^n]$
which leaves node $0$ unmatched, find another unmatched node.

\item
\textsc{Localopt} is the $\TFNP$ problem:
given circuits for a set $S \subseteq [2^n]$ with $2^n-1 \in S$, and a function~$N$ on~$[2^n]$,
find $a \in S$ such that either $N(a) \notin S$ or $N(a) \ge a$.

$\PLS$ is the set of $\TFNP$ problems reducible to {\sc Localopt}.
\end{enumerate}
\end{definition}

Some of these must be stated carefully for them to make sense as $\TFNP$ problems;
for example, in {\sc Leaf} we need to specify how   a general
circuit gives rise to a graph that is guaranteed to have to have degree at most~$2$. 
We will not go into the details of this here.
{\sc Left pigeon} is studied
and observed to be complete for $\PPADS$ in~\cite{buss2012propositional};
notice that~$g$ plays a role like the left-inverse of~$f$.
We give {\sc Localopt} in a slightly different form from how it is introduced 
in~\cite{johnson1988easy}, but it is easily seen to be equivalent.
{\sc Lonely} is equivalent to {\sc Leaf} and is thus also complete for $\PPA$.

We have that all $\TFNP$ subclasses contain $\FP$;
that $\PPP \supseteq \PPADS \supseteq \PPAD$;
that $\PPA \supseteq \PPAD$; and that, at least 
relative to suitable oracles, no other inclusions hold
among the classes in 
Definition~\ref{def:standard_TFNP}~\cite{beame1995relative, 
morioka2001classification,goos2022separations}.
We will also work with $\TFS$ problems, in which the relation 
$R(x,y)$ defining the problem is $\coNP$ rather than polynomial time,
and where we will use a different notion of reducibility.
We postpone the definitions to the beginning of Section~\ref{sec:counterexamples}.

\paragraph{A formal language.}
To write axioms describing complexity classes,
it is convenient to work with a formal language for expressing the properties
of polynomial-time machines. We use a standard approach
from bounded arithmetic. 
We take $L_{\PV}$ to be a language for first-order logic,
whose variables range over binary strings, and which has symbols
for all functions and relations defined by polynomial-time machines~\cite{cook1975feasibly}.
For example it has symbols for equality; for binary arithmetic operations;
and for the function mapping $(C, x, b) \mapsto C_{\restrict x=b}$
where $C, x, b$ represent a circuit, an input gate and a bit.
We will call these \emph{$\PV$ function and relation symbols} 
and call quantifier-free formulas built from them \emph{$\PV$ formulas}.
We evaluate
formulas and sentences in this language as true or false in the natural way.
In particular the truth of a fixed $\PV$ formula can be evaluated in polynomial time
(in the length of the strings assigned to the free variables), 
and below we implicitly associate such a formula with a canonical machine,
constructable from the formula, which carries out this evaluation.

In fact we will usually work with a richer language that talks about
\emph{oracle} machines. Taking $A$ to be a symbol standing for an unspecified oracle,
$L_{\PV(A)}$ is similar to $L_{\PV}$ but now has symbols for $A$ and for
polynomial-time functions and relations defined by machines
equipped with an oracle tape for $A$. We will call these $\PV(A)$ functions
and relations etc.
We will extend this without comment to situations with more than one oracle symbol 
and to oracles which output strings rather than single bits
(the strings will only ever be polynomially long).

If $R(\vec x)$ is a $L_{\PV(A)}$ formula, we will  write expressions like
 $R(\vec x; A)$ if we want to emphasize how the formula depends on queries to $A$.
 Given a concrete oracle $\alpha$, we will write $R(\vec x; \alpha)$ for the predicate
 on strings $\vec x$ defined by evaluating this formula, interpreting the symbol $A$ as a call to $\alpha$.
 We say that an $L_{\PV(A)}$ sentence $R(A)$ is true for~$\alpha$, or that $\alpha \vDash R(A)$,
 if $R(\alpha)$ is true.
 
 \begin{definition}
 A \emph{partial} oracle is one that may be undefined on some inputs. 
 A \emph{sparse} oracle is a finitely-supported partial oracle, encoded as a string listing
all its (input, value) pairs.
 \end{definition}

We can evaluate simple $L_{\PV(A)}$ formulas over partial oracles, if
they are defined in enough places.

\begin{definition}
Let $R(\vec x; A)$ be a $\PV(A)$-formula  (that is, a quantifier-free $L_{\PV(A)}$ formula)
and $\vec s$  a tuple of strings. Let $\alpha$ be a partial oracle.
We write $\alpha \vDash R(\vec s; A)$ 
to express that $\alpha$ is defined at every point where it is queried during the computation of
$R(\vec s; \alpha)$, and that $R(\vec s; \alpha)$ evaluates to true in this computation.
\end{definition}

For fixed $R$, over strings~$\vec s$ and \emph{sparse} oracles~$\alpha$
(recalling that a sparse oracle is encoded as a string),
the predicate $\alpha \vDash R(\vec s; A)$ of $(\vec s, \alpha)$
is  polynomial-time.

\begin{definition}
Let $R(\vec x, \vec y; A)$ be a $\PV(A)$-formula,
$\vec s$ a tuple of strings, and $\alpha$ a partial oracle.
We write
$\alpha \vDash \exists \vec y \,  R(\vec s, \vec y; A)$,
or `` $\exists \vec y \, R(\vec s, \vec y; A)$ is true in~$\alpha$'',
to express that $\alpha \vDash R(\vec s, \vec t ; A)$ for some tuple of strings~$\vec t$.
We call $\vec t$ a \emph{witness}.
\end{definition}

As examples of how we will use this language, 
consider the  sentences
\begin{enumerate}
\item
$\forall C  \forall x, \
x \in \inp(C) \longrightarrow
A(C) = A(C_{\restrict x = 0}) \oplus A(C_{\restrict x = 1})$
\item
$\forall C, \
\inp(C) = \emptyset \longrightarrow A(C) = \eval(C)$
\end{enumerate}
where $C$ is understood as coding circuits, $x$ as coding gates,
and $\oplus$ as addition modulo~2.
These sentences should be read as asserting some properties of the undefined
oracle symbol~$A$. In particular, for any oracle $\alpha$ for
which both sentences are true, we have that $\alpha(C)$ is the 
parity of the size of the set of inputs accepted by circuit~$C$. These are examples
of \emph{universal axioms} for parity.

\begin{definition}
A \emph{universal axiom} is a universal $L_{\PV(A)}$-sentence,
that is, a sentence $\forall \vec x R(\vec x; A)$ in the form of a $\PV(A)$-formula preceded by 
a block of universal quantifiers, which we think of as expressing
some property of the oracle~$A$.
We will often call these simply \emph{axioms}, and may
omit the quantifier block when writing them.
We say the axiom \emph{holds} for an oracle~$\alpha$,
or \emph{is an axiom for $\alpha$}, if 
 $\forall \vec x R(\vec x; \alpha)$ is true.

For a partial oracle $\alpha$, a \emph{witness that the axiom fails} for $\alpha$,
also called a \emph{failure} of the axiom in $\alpha$, is a witness $\vec t$ 
that $\alpha \vDash \exists \vec x \neg R(\vec x; A)$.
\end{definition}

Note that for a partial oracle $\alpha$ it is possible for an axiom not to
hold for $\alpha$, but for there to be no witness that it fails,
simply because $\alpha$ is not defined in the relevant places.

A standard way to give evidence that classes may be separate is to show
that they are separate relative to some oracle. 
The symbol $A$ introduced above will usually not play
this role; we typically use it as a ``fake'' oracle tape that we use to model interactions
with an adversary. There are some subtleties to relativizing 
problems involving evaluating circuits, which we address at the beginning
of Section~\ref{sec:bounded}.

\section{Inspector-Adversary games} \label{sec:SO_setup}

In this section $Q(x,y)$ is a TFNP problem and $\Gamma$ is a finite list of universal axioms.

\begin{definition} \label{def:IA_game}
The \emph{$(\Gamma, Q, x)$-game} is played between an Inspector, who
wants to find~$y$ such that~$Q(x,y)$,
and an Adversary, who claims to know a function $A$ satisfying
all axioms in $\Gamma$. 
At round~$i$, the Inspector produces a string~$q_i$ and 
asks the Adversary ``what is $A(q_i)$?''
The Adversary gives a reply $r_i$, and must give the same reply if the same
query reappears later.

The Inspector wins as soon as she produces either
a string $y$ such that $Q(x,y)$, or
a witness that some axiom in $\Gamma$ fails
in the sparse oracle $\alpha := \{ (q_0, r_0), \dots, (q_{i-1}, r_{i-1}) \}$
formed from the Adversary's replies so far.
\end{definition}

\begin{definition} \label{def:Gamma_protocol}
We say that $Q$ is \emph{solvable over $\Gamma$} if there is a 
polynomial-time strategy for the Inspector to win the $(\Gamma, Q, x)$-game. 
Formally, such a strategy is a polynomial-time algorithm, taking
input $x$ and equipped with an oracle tape labelled~$A$, such that: for
any $x$ and for any oracle used to answer queries to~$A$, the algorithm either
outputs $y$ such that $Q(x,y)$, or outputs a witness that some axiom in $\Gamma$
fails in $\alpha$, where $\alpha$ is the sparse oracle formed from the queries made to~$A$ and replies received
during the run of the algorithm, as in Definition~\ref{def:IA_game}.
\end{definition}

The Inspector's strategy in Definition~\ref{def:Gamma_protocol}
is uniform, in the sense that it is given by a fixed polynomial-time machine
with one parameter~$x$. It will also be helpful for us to have a less-uniform
(but still succinct)
notion of such a strategy, as a self-contained object which can be coded by a string.
We could use circuits with oracle gates for $A$, but we prefer the following.

\begin{definition}
A \emph{protocol} is a pair $(M, 1^t)$ of a description of a machine $M$ with an oracle
tape labelled~$A$, and a time bound $t$ encoded in unary.
\end{definition}

We may think of $M$ as starting with some string on its tape, which is
hard-coded into the description of~$M$.
We could now restate Definition~\ref{def:Gamma_protocol} as:
$Q$ is solvable over $\Gamma$ if there is a polynomial-time function~$f$
such that, for every $x$, $f(x)$ is a protocol which is a winning strategy for the Inspector
in the $(\Gamma, Q, x)$-game.

We can now give the main definition of this section.

\begin{definition} \label{def:Gamma_TFNP}
$\loc\Gamma$ is the  $\NP$ search problem:
given a protocol $(M, 1^t)$,
find a run of machine $M$ for time $t$,  coded as a
list $\alpha$ of queries made to $A$
and replies received, such that the output of the machine
does not witness the failure of any axiom of~$\Gamma$ on $\alpha$.

Equivalently: given a protocol, find a sparse
oracle on which the Inspector, using this protocol, does not find any failure of~$\Gamma$.
\end{definition}

Note that this problem is easy for protocols which make very few queries.
In particular, if a protocol makes no queries at all, then it is unlikely to find a failure of~$\Gamma$,
since there is no support in $\alpha$ for the failure to be defined on
(an exception would be if $\Gamma$ is, say, ``$0{=}1$'').
However we will usually be in an adversarial situation where the 
Inspector does not skip any obvious queries she could use to find a failure.

The next two propositions show that if $\Gamma$ are axioms for some total function $\alpha$,
then $\loc \Gamma$ is a TFNP problem and is complete for the class
of TFNP problems $Q$ which are solvable over $\Gamma$.

\begin{proposition} \label{pro:loc_totality}
If $\Gamma$ holds for some total oracle $\alpha$ then $\loc \Gamma$ is
total.
\end{proposition}

\begin{proof}
Run the protocol and use $\alpha$ to reply to queries to $A$.
The protocol will not find any failure, since there is none.
\end{proof}

\begin{proposition} \label{pro:solvable}
$Q$ is solvable over $\Gamma$ if and only if it is
reducible to $\loc\Gamma$.
\end{proposition}

\begin{proof}
First suppose that $Q$ is solvable over $\Gamma$. By Definition~\ref{def:Gamma_protocol},
from any input~$x$ we can produce in polynomial time a protocol $\pi$ for 
the Inspector which is a winning strategy in the $(\Gamma, Q, x)$-game,
meaning that either it outputs $y$ such that $Q(x,y)$, or it outputs
a witness that some axiom fails in the oracle constructed during the protocol.
We take $\pi$ as our input to $\loc\Gamma$. 
Any solution $\alpha$ is then a sparse oracle on which $\pi$ 
does not find a witness that any axiom fails. Therefore $\pi$,
run on~$\alpha$, must output $y$ such that $Q(x,y)$,
and thus we can solve $Q$ by simulating $\pi$.

For the other direction it is enough that $\loc\Gamma$ is solvable
over $\Gamma$, which is immediate from the definitions.
\end{proof}

We finish the  section by discussing connections with proof 
complexity. 
The Prover-Adversary game introduced in~\cite{pudlak1995lie} is a kind
of interactive proof system, where a defendant (Adversary) claims that an
unsatisfiable propositional formula is true,
and a prosecutor (Prover) tries to catch them out in a simple contradiction,
by asking them to evaluate an adaptive sequence of formulas.
A proof in this setting is a winning strategy for the prosecutor,
which can formally be seen as a treelike proof in a certain system.
Our Inspector-Adversary game resembles a succinct version of this,
where a protocol serves as a succinct description of an exponential-size strategy
by a polynomial-time algorithm, and queries to $A$ can be used
to succinctly encode formulas (e.g. as QBFs). We can also
think of a protocol as a succinct treelike proof using axioms from $\Gamma$;
the proof of Lemma~\ref{lem:Gamma_to_leaf} below can be understood
as formally constructing such a proof (in a form of polynomial calculus).
An approach rather like this was used in~\cite{buss2004bounded} for an alternative proof of the witnessing theorems
for $S^i_2$ and $T^i_2$, in which a first-order proof
is first converted into a succinctly-described exponential-size propositional proof,
and then the witnessing algorithm works by traversing this object.

We can also understand the problems reducible to $\loc\Gamma$ 
as the problems that $\Gamma$ implies are provably total, over a certain
bounded arithmetic theory. We record this as a lemma, but we will
not use it in anything that follows. Let $T$ be the first-order
theory consisting of every universal sentence in the language $L_{\PV(A)}$
which is true for every oracle $\alpha$. This is an ``enriched'' version
of the theory $\PV$ of~\cite{cook1975feasibly}; compare the use of the
theory $\forall \PV (\NN)$ for handling TFNP reductions in~\cite{kol2022approx}.

\begin{lemma} \label{lem:Gamma_over_PV}
Let $R \in \TFNP$. Then $R \le \loc\Gamma$ if and only if $R$ is provably
total in $T+ \Gamma$, that is, if and only if
$T + \Gamma \vdash \forall x \exists y R(x,y)$.
\end{lemma}

\begin{proof}
Suppose that $R \le \loc\Gamma$.
Then $T$ contains the theory $\PV(A)$, a relativized version of $\PV$,
which is enough to prove that any polynomial-time protocol querying~$A$
has a computation. The assumption $\Gamma$ naturally guarantees that these
replies do not falsify any axiom of $\Gamma$. So this establishes
that $T + \Gamma$ proves that $\loc\Gamma$ is total.
Now observe that $R \le \loc\Gamma$
can be written as a true universal $L_{\PV}$ sentence
describing the property of the reduction functions. Therefore this sentence is in $T$,
and we conclude that $T + \Gamma$ also proves that $R$ is total.

For the other direction, suppose 
$T + \Gamma \vdash \forall x \exists y R(x,y)$. By appealing
to compactness and then moving  things to the right, we get that
\[
\forall x, \
 \exists \vec u \, \neg \sigma(\vec u)
 \vee
 \exists \vec v \, \neg \tau(\vec v)
 \vee
 \exists y \, R(x,y)
\]
is provable in pure logic,
where $\forall \vec u \, \sigma(\vec u)$ is the conjunction of some sentences in $T$
and $\forall \vec v \, \tau(\vec v)$ is the conjunction of some sentences in $\Gamma$.
By Herbrand's theorem, and observing that we are able to do definition by cases
in an $L_{\PV(A)}$-term, we can find an $L_{\PV(A)}$-term $t(x)$ such that, for every $x$ and every oracle $\alpha$, the output of $t(x)$ is a triple $(\vec u, \vec v, y)$ satisfying
 either $\neg \sigma(\vec u)$, or $\neg \tau (\vec v)$,
or $R(x,y)$. It is impossible that~$\neg \sigma(\vec u)$, since $\forall \vec u \, \sigma(\vec u)$
is true for every oracle. Therefore $t(x)$ can be read as a protocol for the Inspector,
which queries $A$ and either finds a failure of $\Gamma$ or a solution to $R$.
Hence $R$ is solvable over $\Gamma$.
\end{proof}

\section{Parity} \label{sec:parity}

Let $\Gamma_\oplus$ be the pair of axioms for parity that appeared as
an example in Section~\ref{sec:preliminaries}. We repeat them here,
this time suppressing the universal quantifiers over $x$ and $C$:
\begin{enumerate}
\item
$x \in \inp(C) \longrightarrow
A(C) = A(C_{\restrict x = 0}) \oplus A(C_{\restrict x = 1})$
\item
$\inp(C) = \emptyset \longrightarrow A(C) = \eval(C)$.
\end{enumerate}
Specifically these are axioms for the function, which we will call $\pty$,
mapping a circuit to the parity of the size of the set of accepting inputs.

\begin{theorem} \label{the:Gpty_PPA}
$\loc\Gpty$ is complete for $\PPA$.
\end{theorem}

We show the two directions of this as Lemmas~\ref{lem:lonely_to_Gamma} and~\ref{lem:Gamma_to_leaf} 
below, using two standard $\PPA$-complete problems
{\sc Lonely} and {\sc Leaf} introduced in Definition~\ref{def:standard_TFNP}.
Later, in Theorem~\ref{the:robust_PPA}, we will show that this 
construction of $\PPA$ from $\pty$ does
 not depend on our particular choice
$\Gpty$ of axioms for $\pty$.

To help us describe complicated protocols we introduce a
notion of when a formula is \emph{enforceable} against the Adversary.
Let $R(\vec y)$ be a $\PV(A)$ formula and~$\vec s$ a tuple of strings.
Suppose we are  the Inspector in the $(\Gamma, Q, x)$ game.
We say that $R(\vec s)$ is \emph{enforceable} over $\Gamma$ if
we can force the Adversary in polynomial time either to say that 
$R(\vec s)$ is true, or to lose the game by violating some axiom.
More formally, this means that we have a polynomial-time\footnote{
``Polynomial time'' here strictly means in the length of $x$, but usually
means in the length of $\vec s$ if our argument could take place in any game and does not use $x$,
as in the next lemma.}
 substrategy of the following kind.
We start by asking the Adversary to evaluate $R(\vec s)$,
which may take several individual queries to $A$.
 If the Adversary's replies  evaluate $R(\vec s)$ to true, we are done.
Otherwise, our substrategy must provide us with a series of queries
that inevitably lead the Adversary to falsify some axiom from $\Gamma$, in a way that
gives us a concrete witness to its failure.

\begin{lemma}[well-definedness] \label{lem:parity_well_defined}
For any circuits $C$ and $D$ with the same inputs and computing the same function, 
it is enforceable over $\Gamma_\oplus$ that $A(C) = A(D)$.
\end{lemma}

\begin{proof}
List $\inp(C)$ as $x_1, \dots, x_n$.
If the Adversary tells us that $A(C) \neq A(D)$, then by binary search
using axiom~1 of $\Gpty$
we can find a complete restriction $\vec a$ for which the Adversary tells us that
$A(C_{\restrict \vec x = \vec a}) \neq A(D_{\restrict \vec x = \vec a})$.
But by assumption $\eval(C_{\restrict \vec x = \vec a}) = \eval(D_{\restrict \vec x = \vec a})$,
so this contradicts axiom~2 of $\Gpty$.
\end{proof}

In the light of  Lemma~\ref{lem:parity_well_defined} it is not important to specify
which particular circuit we send to the Adversary when we make a query.
So for a set $X$, implicitly defined by a circuit $C$, we
allow the notation $A(X)$ for a query instead of $A(C)$.
In particular we write $A([2^n])$ for $A(C)$, where $C$ is
the circuit which has $n$ inputs but ignores them and always accepts;
and, for $i \in [2^n]$, we write 
$A(\{i\})$ for $A(C_{\restrict \vec x = \vec a})$, where $C$
is as before and $\vec a$ is the binary expansion of $i$. By Axiom~2
of $\Gpty$ the Adversary should say $A(\{i\})=1$ in this case,
since $\eval(C_{\restrict \vec x = \vec a})=1$.

\begin{lemma} \label{lem:parity_disjoint}
For disjoint  sets $X$ and $Y$ 
(assumed to be defined by circuits with the same inputs)
it is enforceable  over $\Gamma_\oplus$ that 
$A(X \cup Y) = A(X) \oplus A(Y)$.
\end{lemma}

\begin{proof}
Suppose $C$ and $D$ are circuits defining $X$ and $Y$. Let us write $C \vee D$
for the circuit defining $X \cup Y$ by combining the outputs of $C$ and $D$ with a $\vee$-gate.
If $A(C \vee D) \neq A(C) \oplus A(D)$ then by binary search we can find bits $\vec a$
assigning a value to every input gate
such that $A((C \vee D)_{\restrict \vec a}) \neq A(C_{\restrict \vec a}) \oplus
A( D_{\restrict \vec a})$
(where we are suppressing the names of the inputs).
Then $\eval((C \vee D)_{\restrict \vec a}) \neq \eval(C_{\restrict \vec a}) \oplus
\eval( D_{\restrict \vec a})$ which is impossible since $X$ and $Y$ are disjoint.
\end{proof}

\begin{lemma} \label{lem:lonely_to_Gamma}
{\sc Lonely} $\le \loc\Gamma_\oplus$.
\end{lemma}

\begin{proof}
This is adapted from  Lemma~18 of~\cite{buss2015collapsing}.
By Proposition~\ref{pro:solvable} it is enough to show that {\sc Lonely} is solvable over $\Gamma_\oplus$,
by describing a winning strategy for the Inspector. 
We are given an instance of {\sc Lonely} in the form of a circuit
defining a matching on~$[2^n]$. As in~\cite{beame1995relative}
we take the circuit as defining a function $N$ on $[2^n]$
and consider distinct nodes~$u$ and~$v$ as matched precisely when
$N(u)=v$, $N(v)=u$, and neither node is~$0$.
Our task is to find an unmatched node different from~$0$.

We first enforce that $A([2^n])=0$. We can do this because if $A([2^n])=1$,
we can use binary search to find $i$ such that $A(\{2i, 2i+1\})=1$.
But we must have $A(\{2i\})=1$ and $A(\{2i+1\})=1$, so this gives a contradiction.

We partition $[2^n]$ into sets $U := \{ x : x$ is unmatched$\}$,
$L := \{ x \notin U : N(x) > x \}$ and $R := \{ x \notin U: N(x) < x \}$
consisting respectively of unmatched elements, left-hand ends of matchings, and right-hand ends of matchings.
We query $A(U)$, $A(L)$ and $A(R)$.
Since $[2^n]$ is the disjoint union of these sets and $A([2^n])=0$, we must get that
either $A(U)=0$ or $A(L) \cup A(R)=1$, which implies $A(L) \neq A(R)$.

If $A(U)=0$, since we know $0$ is in $U$ we can enforce
$A(U \setminus \{ 0 \})=1$. At this point we can use binary 
search to force the Adversary to show us another element of~$U$,
which is a solution to {\sc Lonely}.

Otherwise $A(L) \neq A(R)$. In this case we define 
$L_i = \{ x \in L : x \le i \}$
and $R_i = \{ x \in R : N(x) \le i \}$.
We have $A(L_0)=A(R_0)=0$ since both sets are empty, 
and $A(L_{2^n})=A(L) \neq A(R) = A(R_{2^n})$.
Using binary search we find $i \in [2^n-1]$ with
$A(L_i)=A(R_i)$ but $A(L_{i+1}) \neq A(R_{i+1})$.

There are now three cases. If $i+1 \in U$ then we have found  a solution to {\sc Lonely}.
If $i+1 \in L$ then $i+1$ and $N(i+1)$ are matched with $i+1 < N(i+1)$,
so $L_{i+1} = L_i \ \dot\cup \ \{ i+1\}$ and $R_{i+1} = R_i \ \dot\cup  \ \{ N(i+1) \}$.
Since $A(L_i)=A(R_i)$ we should now be able to enforce that $A(L_{i+1}) = A(R_{i+1})$, a contradiction.
In the last case we have $i+1 \in R$,
so $N(i+1)$ and $i+1$ are matched with $N(i+1) < i+1$, so $i+1$ is already in $R_i$. 
Thus $L_{i+1} = L_i$ and $R_{i+1}=R_i$, giving a contradiction.
\end{proof}

\begin{lemma} \label{lem:Gamma_to_leaf}
$\loc\Gamma_\oplus \le$ {\sc Leaf}.
\end{lemma}

The proof of this lemma is inspired by Theorem~26 of~\cite{buss2015collapsing},
which concerns a theory extending $\PV$ with the ability to count mod~2,
which is essentially the same as our notion of a game over $\Gamma_\oplus$,
via Lemma~\ref{lem:Gamma_over_PV}.
It is shown in~\cite{buss2015collapsing} that
 if a TFNP problem is provably total in this theory, then its 
 propositional translation has low-degree treelike refutations in polynomial
calculus, and hence by~\cite{buss1996proof}  has low-degree
Nullstellensatz refutations;
see also Lemma C3 of~\cite{de2020lifting}. It was shown in~\cite{goos2019adventures} that 
the existence of a sufficiently uniform low-degree Nullstellensatz refutation is equivalent to the
original problem being in $\PPA$.  We give a self-contained proof.

\begin{proof}
We are given a protocol $(M, 1^t)$ for the Inspector. We think of this protocol
as defining a binary tree $T$ of depth $t$, where each node is labelled with the replies made 
by the Adversary so far, and where each branch stops as soon as the Inspector wins and otherwise
carries on until the protocol is finished.
We want to find a leaf whose label does not violate any axiom of $\Gamma_\oplus$.

By appealing to the existence of a universal circuit, we may assume that
there is one fixed circuit~$C$, with a polynomial number $n$ of input
gates, such that every
query to the Adversary in the tree is about a restriction of~$C$.
 For each input~$\vec a$ to~$C$ we introduce a formal variable
$x_{\vec a}$. We define an assignment $\beta$
by $\beta(x_{\vec a}) = \eval( C_{\vec a})$. 
We can now write every query made by the Inspector as a formal sum of a 
set of variables, where the set is polynomial-time computable 
from the history of the game so far.
Precisely, we write the query $A(C_{\restrict \vec e})$ as $\sum \{ x_{\vec a} : \vec a$ extends~$\vec e \}$.

Each node $u$ in $T$ is now labelled with a conjuction $K_u$ of assertions made so far by the Adversary.
$K_u$ has
the form $S_1 = b_1 \, \wedge \, \dots \, \wedge  \, S_k = b_k$ where $k \le t$,
each~$S_i$ is a sum of variables and each $b_i$ is $0$ or $1$.
We replace $K_u$ with the polynomial
$p_u := \prod_{i=1}^k (S_i +(1 - b_i))$ over $\FF_2$, so that the equality $p_u = 0$
asserts the \emph{negation} of $K_u$. 
We expand $p_u$ as a sum of monomials, without doing any cancellation.
We label the root of the tree with the polynomial $1$
(that is, the empty product).
We say that a monomial is \emph{active} if it evaluates to $1$ under
the assignment~$\beta$.

Our goal now is to define an instance of {\sc Leaf}, 
in the form of a graph~$G$
whose vertex-set
consists of all occurrences of active monomials in $T$. The distinguished 
$G$-leaf in the instance 
(that is, the leaf at node $0$ in the definition of {\sc Leaf} in 
Definition~\ref{def:standard_TFNP})
will be the monomial $1$ at the root of~$T$, and we will define
the graph so that, if we find any other $G$-leaf, it must be on a leaf of $T$ 
which corresponds to a solution of $\loc\Gamma_\oplus$. 
We emphasize that formally the input to {\sc Leaf} is not~$G$ itself but rather a circuit describing~$G$, where the
circuit is constructible in polynomial time from the protocol $(M,1^t)$. 
We will not give details of how this circuit is constructed but it will 
be clear that this can be done, since $G$ is defined in a very local way from~$T$.

We divide the leaf nodes of $T$ into three classes.
For a leaf node $u$, if $K_u$ falsifies
axiom 1, then $u$ is \emph{inconsistent}; if $K_u$ falsifies axiom 2,
then $u$ is \emph{incorrect}; otherwise $u$ is \emph{good}. 
Now for each non-leaf, inconsistent or incorrect node $u$ of $T$ we will define a matching $G_u$. 
For non-leaf nodes~$u$, $G_u$ will be a matching on all the active monomials 
in $p_u$, $p_v$ and $p_w$, where~$v$ and~$w$ are the parents of $u$.
For inconsistent or incorrect nodes~$u$, $G_u$ will be a matching on all active monomials in~$p_u$.
For good nodes~$u$, $G_u$ will not be a matching but will be a 
degenerate graph 
whose vertices are the active monomials in $p_u$, and which does not have any edges.
We take $G$ to be the graph whose vertices are all active monomial occurrences in $T$,
and whose edges are the discrete union of the edges in $G_u$ over all $u \in T$ (we allow repeated edges).
Thus in $G$, with the exception of the active monomial~$1$
at the root of $T$ and the active monomals at good leaves~$u$, every vertex has degree~$2$,
because it is matched in precisely two of the graphs $G_v$
over $v \in T$. 
Thus any solution to this instance of {\sc Leaf} will be a monomial occurrence
at a good node~$u$, which in particular finds us a good~$u$, as required.

We now describe the graphs $G_u$.
First suppose $u$ is a non-leaf node. Then by construction $p_u$ is the sum $p_v + p_w$, where~$v$ and~$w$ are the parents of $u$.
Any monomial occurrence in~$p_v$ (or $p_w$)
either appears also in~$p_u$ or is cancelled against
an occurrence of the same monomial in~$p_w$ (or $p_v$).
Thus there is complete matching~$G_u$ of all occurrences of active monomials at the three nodes $u$, $v$ and $w$
which can be described by a small circuit.

Suppose $u$ is an inconsistent leaf. Then there exist sums $S_i$ and $S_j$
such that~$K_u$ contains either three assertions $S_i = 0$, $S_j=1$ and $S_i + S_j = 0$
or three assertions $S_i = 1$, $S_j=1$ and $S_i + S_j = 1$. We consider the first case;
the second one is similar.
We have that $p_u$ is a product 
$(S_i +1) S_j (S_i + S_j +1) q$
for some $q$. We can rearrange this as
$
(S_i +1) S_i \cdot S_j q
\ 
+
\ 
(S_j +1) S_j \cdot (S_i+1) q.
$
If we expand $(S_i+1)S_i$ into a sum of monomials, every active monomial
appears an even number of times, and in particular we can 
construct a circuit describing a complete matching on this sum which matches each active monomial with another occurrence of the same monomial.
We can do the same for $(S_j +1) S_j$, and from these two matchings we can define
a matching~$G_u$ of the same kind on the whole expression~$p_u$. 

Suppose $u$ is an incorrect leaf. Then for some single variable $x_{\vec a}$, either
$\beta(x_{\vec a})=0$ and $K_u$ asserts $x_{\vec a} = 1$, 
or $\beta(x_{\vec a})=1$ and $K_u$ asserts $x_{\vec a} = 0$.
In the first case, $p_u$ has a factor $x_{\vec a}$ which is $0$ under~$\beta$.
Thus $p_u$ has no active monomials, and there is nothing to do. 
In the second case, $p_u$ has a factor $x_{\vec a} + 1$ which 
becomes $1+1$ under~$\beta$, allowing us to define a circuit
describing a complete matching~$G_u$ on the 
active monomials of~$p_u$. 
This completes the construction of~$G$.
\end{proof}

\subsection{Other choices of axioms} \label{sec:robustness}

We would like to show that if you apply our machinery to the  decision class $\pty \P$,
then the TFNP subclass $\PPA$ arises in a canonical way as the result.
So far we have not really shown this;
we rather showed in Theorem~\ref{the:Gpty_PPA} that $\PPA$ arises from a particular 
(although very natural)
choice of axiomatization~$\Gamma_\oplus$ for the function~$\pty$.
We now show something much closer to what we want:
\begin{itemize}
\item
Let $\Gamma$ be the axiomatization of $\pty$ arising from any recursive definition of~$\pty$
(of a certain form).
Then $\loc\Gamma$ is complete for $\PPA$.
\end{itemize}
This is shown formally below as Theorem~\ref{the:robust_PPA}. We first
describe the kind of recursive definition which we allow.

\begin{definition}
A \emph{$k$-query black-box recursive definition of $\pty$} is an
axiom of the form $\forall C, A(C) = \theta(C ; A)$ which holds for $\pty$, where 
$\theta$ is a $\PV(A)$-formula of a restricted form which we call a 
\emph{$k$-query black-box evaluator.}

Formally, a $k$-query black-box evaluator is a polynomial-time
oracle machine which  has only limited access to $C$ and $A$.
 Specifically it
can query $\inp(C)$  and then adaptively make at most $k$ queries of the form
\begin{enumerate}
\item
query $A(C_{\restrict \rho})$ where $\rho$ is a restriction which sets
exactly one input gate, or
\item
ask for $\eval(C_{\restrict \rho})$ under a restriction $\rho$ to all input gates.
\end{enumerate}
It has no other access to $A$ or $C$. It finally outputs $0$ or $1$.
\end{definition}

Recall that ``holds for $\pty$'' means simply that the sentence
$\forall C, \pty(C) = \theta(C ; \pty)$ is true.
As an example of a 2-query black-box recursive definition of $\pty$,
let $\theta$ be the procedure:  if $C$
has no inputs  return $\eval(C)$; otherwise return 
$A(C_{\restrict x =0}) \oplus A(C_{\restrict x=1})$, where $x$
is the first input gate of $C$. 

\begin{theorem} \label{the:robust_PPA}
Let $\Gamma$ be a black-box recursive definition of $\pty$ with constantly
many queries.
Then $\loc\Gamma \equiv \loc \Gpty$ and thus $\loc\Gamma$ is complete for $\PPA$.
\end{theorem}

Theorem~\ref{the:robust_PPA} is a consequence of 
Lemma~\ref{lem:robust_parity_queries} below,
which lets us convert failures of $\Gpty$ into failures of $\Gamma$ and vice versa.
 We first state a simple
lemma about when a small number of parity queries get ``consistent'' replies.

\begin{lemma} \label{lem:parity_consistency}
Let $\theta$ be a $k$-query black-box evaluator. There is a constant $n_k$ such that the following is true, for 
any oracle~$\alpha$.
Suppose $C$ is a circuit with at least $n_k$ inputs, and 
the queries made in the computation of $\theta(C; \alpha)$ satisfy the minimal consistency condition that:
there is $b \in \{0,1\}$ such that,
whenever both 
$\alpha(C_{\restrict x =0})$
and 
$\alpha(C_{\restrict x =1})$ 
were queried for $x \in \inp(C)$, 
then $\alpha(C_{\restrict x =0}) \oplus \alpha(C_{\restrict x =1}) = b$.
Then there is a circuit $Z$ with $\inp(Z)=\inp(C)$ such that $\pty(Z) = b$
and $\theta(Z; \pty)$ has the same output as $\theta(C; \alpha)$.
\end{lemma}

\begin{proof}
It is clear that, for large enough $n_k$, we can define a set $Z$ 
with $\pty(Z)=b$, whose restricted
 parities match  the oracle replies made in the computation, and which 
also contains or omits elements to match the queries made to $\eval(C_{\restrict \rho})$. 
We then simply write $Z$ as a circuit --- we do not care about its size.
\end{proof}

For the rest of this section we will say that
 a circuit $C$ is \emph{well-behaved} with respect to an oracle~$\alpha$
if $\alpha(C) = \alpha(C_{\restrict x=0}) \oplus \alpha (C_{\restrict x = 1})$
for every $x \in \inp(C)$,
and $\alpha(C) = \eval(C)$ if $C$ has no input gates. 
Otherwise we call $C$ \emph{badly-behaved}.
(Notationally, it is convenient to associate well-behavedness with the circuit,
rather than with the oracle, because we will discuss several circuits
but one fixed oracle.)

\begin{lemma} \label{lem:robust_parity_queries}
Let $\Gamma$ be a $k$-query black-box recursive definition of $\pty$
of the form $\forall C, A(C) = \theta(C ; A)$,
for some $k \in \NN$.
 Let $\alpha$ be any oracle.
Then given any circuit~$C$ with $\alpha(C) \neq \theta(C ; \alpha)$ we can 
find in polynomial  time a restriction $C_{\restrict \rho}$ of $C$ which is
badly-behaved.
Conversely given any circuit $D$ which is badly-behaved we can find 
in polynomial time a restriction $D_{\restrict \rho}$ of $D$ 
such that $\alpha(D_{\restrict \rho}) \neq \theta(D_{\restrict \rho} ; \alpha)$.
\end{lemma}

\begin{proof}
Suppose first we have $C$ with $\alpha(C) \neq \theta(C ; \alpha)$.
If $A$ has fewer than $n_k$ inputs, query $\alpha(C_{\restrict \rho})$ for every 
restriction $\rho$, of any size. One of the $C_{\restrict \rho}$ must be badly-behaved, since otherwise we
would inductively have $\alpha(C_{\restrict \rho})=\pty(C_{\restrict \rho})$ for every~$\rho$,
which implies $\pty A \neq \theta(A; \pty)$, contradicting the assumption
that $\Gamma$ holds for $\pty$.

Otherwise $C$ has $n_k$ or more inputs, and we claim in this case
that $C$ itself is badly-behaved.
For a contradiction, suppose it is well-behaved.
Consider the run of the algorithm $\theta(C; \alpha)$. 
By the well-behavedness assumption, this computation meets
the consistency condition in Lemma~\ref{lem:parity_consistency},
with $b = \alpha(C)$.
Therefore by the lemma there is a circuit $Z$ such that $\pty(Z) = \alpha(C)$ and $\theta(Z; \pty) = \theta(C; \alpha)$.
It follows that 
$\pty(Z) = \theta(Z ; \pty)$, again contradicting that $\Gamma$ holds for $\pty$.

The other direction uses a similar idea. 
We are given $D$ that is badly-behaved. 
We may assume without loss
of generality that, for every $y \in  \inp(D)$, both $\alpha(D_{\restrict y=0})$
and $\alpha(D_{\restrict y = 1})$ are well-behaved. This is because otherwise
we could replace $D$ with one of these restrictions, and then repeat the argument;
in other words, we can find a  ``locally minimal'' badly-behaved $D$ 
in this sense in polynomial time.

If $D$ then has fewer than $n_k$ inputs, we query $\alpha(D_{\restrict \rho})$ for
every possible restriction~$\rho$. We will necessarily find some $\rho$ such that $C := D_{\restrict \rho}$ 
is badly-behaved but every further restriction of $C$, by any number of inputs, is well-behaved;
in this sense $C$ is ``globally minimal''.
It follows that for every non-trivial restriction~$\sigma$, 
we have $\alpha(C_{\restrict \sigma})=\pty (C_{\restrict \sigma})$,
and thus that $\theta(C; \alpha) = \theta(C; \pty)$.
Also $\theta(C; \pty) = \pty C$, since $\Gamma$ holds for $\pty$.
Therefore $\theta(C; \alpha) = \pty C$.
On the other hand, 
since $C$ is badly-behaved, it has an input gate $y$
such that $\alpha(C) \neq \alpha(C_{\restrict y=0}) \oplus \alpha(C_{\restrict y=1})
= \pty (C_{\restrict y=0}) \oplus \pty (C_{\restrict y=1}) = \pty C$.
Hence $\alpha(C) \neq \theta(C ; \alpha)$.
 
Otherwise $D$ has $n_k$ or more inputs, and we will show that in this
case the local minimality of $D$ already implies $\alpha(D) \neq \theta(D; \alpha)$. 
Let input gate $x$ witness the bad behaviour of $B$, so that 
$\alpha(D) \neq \alpha(D_{\restrict x=0}) \oplus \alpha(D_{\restrict x = 1})$.
We first observe that for every input gate $y$ different from $x$, we have
\begin{equation} \label{eq:ok_behaved}
\alpha(D_{\restrict y=0}) \oplus \alpha(D_{\restrict y = 1})
=
\alpha(D_{\restrict x=0}) \oplus \alpha(D_{\restrict x = 1}).
\end{equation}
This is because by well-behavedness of $D_{\restrict y=0}$
and $D_{\restrict y = 1}$ we can expand the left-hand sum
as 
\[
\alpha(D_{\restrict y=0, x=0}) \oplus \alpha(D_{\restrict y = 0, x=1})
\oplus
\alpha(D_{\restrict y=1, x=0}) \oplus \alpha(D_{\restrict y = 1, x=1})
\]
which is the same as what we can get expanding the right-hand sum, 
since restricting first $y$ and then $x$ gives the same circuit as restricting
first $x$ and then $y$. Now consider a run of $\theta(D; \alpha)$.
If we set $b = \alpha(D_{\restrict x=0}) \oplus \alpha(D_{\restrict x = 1})$,
then by observation~(\ref{eq:ok_behaved}) this meets
the consistency condition in Lemma~\ref{lem:parity_consistency}.
Thus we can find $Z$ with $\pty(Z) = b \neq \alpha(D)$
and $\theta(Z; \pty) = \theta(D; \alpha)$.
By the assumption that $\Gamma$ holds for $\pty$
we have $\pty(Z) = \theta(Z ; \pty)$, and therefore
$\alpha(D) \neq \theta(D; \alpha)$ as required.
\end{proof}

\section{PSPACE} \label{sec:pspace}

\subsection{TQBF and PSPACE computation}

A QBF $F$ consists of a circuit $C$ preceded by a (possibly empty) string of universal or existential
Boolean quantifiers, each associated with a distinct input gate of $C$. We will
write QBFs as $\forall x \exists y C$, $\exists {x_1} \exists {x_2} C$ etc.
where as before $x, y, x_1, x_2, \dots$ name input gates, and we interpret
for example $\forall x \dots$  as ``for both assignments to input gate~$x$ \dots''
If $F = \sigma C$ is a QBF, where $\sigma$ represents the string of quantifiers,
we take $\inp(F)$ to be the input gates of $C$ which are not bound in $\sigma$,
and for $x \in \inp(F)$ write $F_{\restrict x=b}$ for the QBF $\sigma (C_{\restrict x=b})$.
If $\inp(C) = \emptyset$ we say $F$ is \emph{closed}.
We can stratify QBFs into classes $\Sigma^q_k$ and $\Pi^q_k$ in the usual way,
based on the outermost quantifier and the number of alternations of quantifier blocks.

Let $\Gqbf$ consist of the following axioms, quantifying over closed QBFs,
which hold for the function TQBF which evaluates closed QBFs as true or false:
\begin{enumerate}
\item
$A(\forall x F) \longleftrightarrow A(F_{\restrict x = 0}) \wedge A(F_{\restrict x = 1})$
\item
$A(\exists x F) \longleftrightarrow A(F_{\restrict x = 0}) \vee A(F_{\restrict x = 1})$
\item
if $F$ is quantifier-free then $A(F) = \eval(F)$
\end{enumerate}
where  $\eval(F)$ in 3 means that we evaluate the circuit in $F$, which has no
inputs since $F$ is closed.

\begin{lemma}[well-definedness] \label{lem:QBF_well_defined}
Let $\sigma C$ and $\sigma D$ be closed QBFs with a common quantifier
prefix~$\sigma$, where $C$ and $D$ have the same inputs and compute the same function.
Then it is enforceable over $\Gqbf$ that $A(\sigma C) = A(\sigma D)$.

Furthermore if we take $\sigma' C'$ to be the dual of $\sigma C$,
with $\forall / \exists$ quantifiers swapped and with the output of $C$ negated,
then it is enforceable that $A(\sigma' C') \neq A(\sigma C)$.
\end{lemma}

\begin{proof}
This is easily proved by binary search, along the same lines as Lemma~\ref{lem:parity_well_defined}.
\end{proof}

Consider the function that takes as input a triple $(C,w,t)$ where $C$ is a circuit with
the same number $n$ of input gates and output gates, $w$ is an $n$-bit string and $t$ is a number
(written in binary); the output is the string $C^t(w)$, that is, the result of applying $C$
iteratively $t$ times to $w$. This function captures the natural definition of PSPACE,
and we will call it simply PSPACE. 
In this section we will write $\eval(C,z)$ for the string obtained
by evaluating a circuit $C$ on a string~$z$.

We take 
$\Gps$ to be the following two axioms for PSPACE, 
which are implicitly universally quantified over $C$, $w$ and $t$:
\begin{enumerate}
\item
$A(C,w,0) = w$ 
\item
$A(C,w,t+1) = \eval(C, A(C,w,t))$.
\end{enumerate}
If we write $A(C,w,t+1)$ as $C^t(w)$ these become the expected equations
 $C^0(w)=w$ and $C^{t+1}(w) = C(C^t(w))$.

\begin{lemma}[well-definedness] \label{lem:pspace_well_defined}
For any circuits $C$ and $D$ computing the same function, 
it is enforceable over $\Gps$ that $A(C,w,t) = A(D,w,t)$.
\end{lemma}

\begin{proof}
We have $A(C,w,0)=A(D,w,0)$, since otherwise axiom 1 fails right away.
If $A(C,w,t) \neq A(D,w,t)$ then by binary search we can find~$i$
such that $A(C,w,i) = A(D,w,i)$ and $A(C,w,i+1) \neq A(D,w,i+1)$,
contradicting axiom~2.
\end{proof}

\begin{theorem} \label{the:pspace_equivalence}
$\loc\Gps \equiv \loc\Gqbf$.
\end{theorem}

The proof amounts to giving a ``feasibly witnessable" version
of the standard argument that TQBF is PSPACE 
complete~\cite{stockmeyer1973word}. We prove each direction 
as a separate lemma.

\begin{lemma}
$\loc\Gqbf \le \loc\Gps$.
\end{lemma}

\begin{proof}
We are given an instance of $\loc\Gqbf$, in the form of
a protocol $\pi$ for the Inspector asking queries to a $\Gqbf$ adversary $A$.
We will describe a strategy for solving this instance in a game
against a $\Gps$ adversary.

We turn QBF queries into PSPACE queries in
the standard way. That is, we recursively define machines $M_k$ to evaluate
QBFs with $k$ quantifiers, either by using two calls to $M_{k-1}$ if $k>0$,
or by simply evaluating the circuit if $k=0$. Taking this circuit
evaluation as the basic step, $M_k$ runs for $2^k$ steps. 

Now we simulate the protocol $\pi$. Each time it makes a query
$A(F)$, where $F$ has~$k$ quantifiers, we query
the result of running $M_k$ on $F$. If in the simulation
we never encounter a failure of $\Gqbf$, then we have solved our
instance of $\loc\Gqbf$.
Otherwise suppose we encounter a failure. For example, 
suppose that
while simulating queries $A(\forall y F)$,
$A(F_{\restrict y=0})$, $A(F_{\restrict y=1})$
the $\Gps$ adversary asserted
$M_k(\forall y F) = \top$,
$M_{k-1}(F_{\restrict y=0})= \top$
and
$M_{k-1}(F_{\restrict y=1})= \bot$.

From the way the machines are defined,
the sequence of steps from $2^{k-1}{+}1$ to $2^k$ of the computation of $M_k(\forall y F)$
resembles a computation of $M_{k-1}(F_{\restrict y=1})$, and since
the final output is~$\top$ this subcomputation must output $\top$,
by the properties of the transition function of~$M_k$ at step~$2^k$. 
Therefore the $\Gps$ adversary is asserting that he knows two
computations of $M_{k-1}(F_{\restrict y=1})$, one with output
$\bot$ and one with output $\top$. We can use binary search to find
the first place these computations differ, which is a failure of $\Gps$.
Other failures of $\Gqbf$ are handled similarly.
\end{proof}

\begin{lemma}
$\loc\Gps \le \loc\Gqbf$.
\end{lemma}

\begin{proof}
We are given an instance of $\loc\Gps$, in the form of
a protocol $\pi$ for the Inspector asking queries to a $\Gps$ adversary $A$.
We will describe a strategy for solving this instance in a game
against a $\Gqbf$ adversary.

We simulate $\pi$. Suppose we get to a query to
$A(C, w, t)$, where $C$ has $n$ inputs and outputs. 
Letting
$k=\log t$, we let $\Reach_0(u,v), \dots,
\Reach_k(u,v)$ be the usual QBFs from the proof that QBF
is PSPACE-hard, where $\Reach_i(u,v)$ is intended to
express (by recursively using $\Reach_{i-1}$)
that $v$ can be reached from $u$ by $2^k$ iterations
of applying $C$.
Then for each $i$ we ask the $\Gqbf$ adversary whether 
$\forall u \exists ! v \Reach_i(u,v)$ is true, by formalizing
it as a suitable QBF. For $i=0$ he must say it is true,
since otherwise it is easy to force him into a contradiction 
directly using $C$; then inductively once he says it is
true for $i$ he must also say it is true for~$i+1$, 
since otherwise again it is straightforward to find
a contradiction.

We find $k_1 > \dots > k_r$ such that 
$t = 2^{k_1} + \dots + 2^{k_r}$, and ask the $\Gqbf$
adversary for the configuration of the machine at each time
$2^{k_1}$, $2^{k_1} + 2^{k_2}$, \dots, $t$ by asking
first for $v_1$ such that $\Reach_{k_1}(w,v_1)$, then 
for $v_2$ such that $\Reach_{k_2}(v_1, v_2)$ and so on; the 
$\Gqbf$ adversary must give us such witnesses or he would
contradict his assertion $\forall u \exists  v \Reach_{k_i}(u,v)$.
Finally we reply to the original query $A(C, w, t)$
with~$v_r$,  corresponding to time~$t$.

Now suppose that while running the simulation
we encounter a failure of $\Gps$. This means that
we responded to some query $A(C,w,t)$ with some $v_r$
as described above, but on the query
$A(C, w, t+1)$, we carried out the same procedure and 
get a result $v'$ which is different from 
the expected answer $C(v_r)$.
At this point the $\Gqbf$ adversary is describing
 two different paths through
the space of configurations, of the same total length
but leading to two different configurations.
From this it is  straightforward to set up
a binary search that will find some $u$ and $i$
such the Adversary must simultaneously
assert $\Reach_i(u,z)$ for two different strings $z$,
contradicting his earlier uniqueness assertion.
\end{proof}

\subsection{PLS and a polynomial hierarchy inside TFNP} \label{sec:PLS_in_PSPACE}

\begin{definition}
For $k \in \NN$, we take $\Gqbf^k$ to be the same as $\Gqbf$, except that we only
require that the axioms hold for QBFs in $\Sigma^q_k \cup \Pi^q_k$.
\end{definition}

The problems
$\loc\Gqbf^k$ form a kind of analogue of the polynomial hierarchy inside
TFNP.  
We show here that {\sc Localopt}, the complete problem for PLS from
Definition~\ref{def:standard_TFNP}, is reducible to $\loc \Gqbf^1$. In 
Proposition~\ref{pro:QBF1_from_PNP} in the next section
we will prove the converse direction, establishing 
that $\loc \Gqbf^1$ is complete for~$\PLS$.\footnote{
This raises an interesting possibility.
It is known that NP is randomly reducible to parity~\cite{valiant1986np}.
We show that {\sc Localopt} is solvable over the game in which
the Adversary claims to answer NP queries. Is it possible
to show that, if we allow some randomness, it is solvable
over the game in which the Adversary claims to answer parity queries?
That is, could we show in this way that $\PLS$ is randomly reducible in some sense to $\PPA$?
}
We will show much more about this hierarchy in Sections~\ref{sec:dt}
and~\ref{sec:bounded}.

\begin{proposition} \label{pro:QBF1}
$\textsc{Localopt} \le \loc \Gqbf^1$ and  hence $\PLS \le \loc \Gqbf^1$.
\end{proposition}

\begin{proof}
Suppose our instance of $\textsc{Localopt}$ lives on the interval $[2^n]$ and
is given by a set $S$ and a neighbourhood function~$N$. 
We know that $2^n-1 \in S$ and we want
to find a ``local minimum'' $a \in S$ such that either $N(a) \notin S$ or $N(a) \ge a$. 

We first ask the Adversary ``is there an element of $S$ in the interval $[0,2^n)$?'' Formally,
we query $A(\exists x_1 \dots \exists x_n C)$ where $C$ is the circuit defining~$S$.
We may assume the Adversary replies $\top$, since otherwise using axiom 2 we could enforce in turn
$A(\exists x_2 \dots \exists x_n C_{\restrict x_1 = 1}) = \bot$, \dots, $A(C_{\restrict x_1=1, \dots, x_n=1})=\bot$
and the last of these violates axiom 3, since $2^n-1 \in S$ so $\eval(C_{\restrict x_1=1, \dots, x_n=1}) = \top$.

By asking similar queries, we imitate how we would look for a globally minimum element of $S$ by doing a binary search 
over $[2^n]$. We will eventually find $a$ such that $a \in S$ and the Adversary has asserted
that there is no element of $S$ in the interval~$[0,a)$; typically this will not be a single assertion,
rather it will be that $[0,a)$ is the union of intervals that the Adversary asserted were empty during the binary search.
If $N(a) \notin S$ or $N(a) \ge a$, we have solved our instance of $\textsc{Localopt}$. 
Otherwise $N(a)<a$ and $N(a) \in S$, so we have found a member of $S$ in an interval
which the Adversary asserted was empty. Using this we can force the Adversary
into an immediate contradiction.
\end{proof}

We mention here an example of what can go wrong if we choose our axioms poorly.
Consider an arbitrary TFNP problem, say {\sc Pigeon}. 
We can write a polynomial-time formula $\phi(F)$ which is true precisely
of QBFs $F$ which have the form ``there is a collision in the function 
$[2^n] \rightarrow [2^n-1]$ defined by $C$'' for some circuit $C$.
Such an $F$ is always a true QBF, so we could take $\Gamma$ to be $\Gqbf^1$
plus the axiom $\forall F, \phi(F) \rightarrow A(F)$ and $\Gamma$
would still be a valid set of axioms for TQBF on $\Sigma^q_1 \cup \Pi^q_1$. 

However, over this $\Gamma$ the problem {\sc Pigeon} is now solvable.
Given a circuit~$C$, the Inspector forms the QBF $F$ above
expressing that $C$ has a collision, and queries~$A(F)$. The Adversary
must either answer ``true, there is a collision'' or lose the game
by violating the new axiom. After that the Inspector can query
the bits of the collision, and the Adversary is forced to either reveal
suitable bits, or lose. In this was we get that 
{\sc Pigeon} $\le \loc\Gamma$, so in particular $\loc\Gamma$ is different
from $\loc\Gqbf^1$, which is in PLS.
We note that this trick relies on axioms being able to access the internal
structure of $F$, which is what we disallow in Lemma~\ref{lem:robust_parity_queries}.

\subsection{PPA, PPP and a class for $\#$P}

In Proposition~\ref{pro:QBF1} we showed that $\PLS \le \loc\Gps$.
We show now that the same is true for the other classical TFNP classes
$\PPA$ and $\PPP$, and hence automatically for their subclasses
PPADS and PPAD.

We will need to ask the Adversary some questions about counting.
We could do this directly, by asking about a natural PSPACE machine that counts,
but instead we use
the machinery we have developed.
We define axioms $\Gsh$
for $\# \P$, that is, for full counting. Again $A(C)$ will 
be a string rather than a single bit, and is intended
to express the number of satisfying assignments of $C$.
We take:
\begin{enumerate}
\item
$x \in \inp(C) \longrightarrow
A(C) = A(C_{\restrict x = 0}) + A(C_{\restrict x = 1})$
\item
$\inp(C) = \emptyset \longrightarrow A(C) = \eval(C)$.
\end{enumerate}
These have analogous well-behavedness properties as were
shown for $\Gpty$ in Lemmas~\ref{lem:parity_well_defined}
and~\ref{lem:parity_disjoint}, by essentially the same proofs.

\begin{theorem}
$\loc\Gsh \le \loc\Gps$.
\end{theorem}

\begin{proof}
We show how to use a $\Gps$ adversary to simulate a protocol $\pi$
which queries a $\Gsh$ adversary. Suppose $A(C)$ is queried in $\pi$,
where $C$ has $n$ inputs.
We let $D$ be a circuit giving the transition function of the
space~$2n$ machine which steps through the inputs 
to~$C$ and increments a counter whenever $C$ accepts the input.
We send $(D, 0, 2^n)$ to the $\Gps$ adversary, read off
the counter in the configuration they reply with, and use this number
as the  reply to $A(C)$. If we also get queries to 
$A(C_{\restrict x = 0})$ and $A(C_{\restrict x = 1})$,
and give replies using this procedure which do not sum to
our reply to $A(C)$, it is straightforward to use binary search
to find an error in one of these PSPACE computations.
\end{proof}

\begin{theorem} \label{the:PPA_Gsh}
$\PPA \le \loc\Gsh$.
\end{theorem}

\begin{proof}
This follows from Theorem~\ref{the:Gpty_PPA}, 
that $\loc\Gpty$ is complete for $\PPA$,
and the trivial simulation $\loc\Gpty \le \loc\Gsh$.
\end{proof}

\begin{theorem} \label{the:PPT_Gsh}
$\PPP \le \loc\Gsh$.
\end{theorem}

\begin{proof}
We sketch a proof that {\sc Pigeon} is solvable over $\Gsh$.
We have a function $f:[2^n] \longrightarrow [2^n-1]$ and
want to find a collision. We set $P_i := \{ x \in [2^n] :
f(x) \le i \}$.
We can enforce that $A(P_0)=0$ and $A(P_{2^n-1})=[2^n]$,
as in the proof of Lemma~\ref{lem:lonely_to_Gamma}.
Then we use binary search to find $i$ such
that $A(P_i) \le i$ and $A(P_{i+1}) > i+1$,
so that there must be a collision at hole~$i$.
\end{proof}

Let us also observe for the record:

\begin{theorem}
$\PLS \le \loc\Gsh$.
\end{theorem}

\begin{proof}
We have that $\loc\Gqbf^1$ is solvable over $\Gsh$,
by using calls to counting to simulate NP queries in the obvious way.
\end{proof}

\sloppypar It is an interesting question whether every problem 
in the sequence $\loc\Gqbf^1$, $\loc\Gqbf^2$, \dots,
our analogue of the polynomial hierarchy,
is reducible to $\loc\Gsh$. Proving this would amount
to carrying out a  proof of Toda's theorem in the game
over $\Gsh$~\cite{toda1991pp}. We do not see any immediate obstacle
to this; in particular the derandomization trick in Toda's theorem
should give a way  to query statements about counting how many
formulas have odd parity, which we can use to express many parts of the proof.
This would also be interesting for propositional proof complexity,
as it would presumably give a new, purely algebraic system equivalent
to constant-depth Frege with counting gates (TC$^0$-Frege).
We note that Toda's theorem is formalized in a kind of bounded arithmetic 
in~\cite{buss2015collapsing}. The theory used there corresponds, in our setting,
to something like APPROX with parities.

\section{TF$\mathbf{\Sigma^p_2}$ and counterexample reducibility} \label{sec:counterexamples}

We move on to our second main construction.
A \emph{$\TFS$ search problem} is specified by a ternary polynomial-time relation
$R(x,y,z)$ and polynomials $p$, $q$ such that for every $x$, there is~$y \in [2^{p(|x|)}]$ such 
that $R(x,y,z)$ holds for all $z \in [2^{q(|x|)}]$.
It represents the search problem of finding such a~$y$, given~$x$.
As with $\TFNP$, we will refer to this problem as $R$ and 
will often not write the bounds on $y$ and $z$ explicitly.
We will usually call~$x$ the \emph{input} or \emph{instance}
and $y$ the \emph{solution};
we will sometimes call $z$ the \emph{counterexample},
especially when $\neg R(x,y,z)$ holds.

We introduce the following notion of reducibility
between two $\TFS$ problems.

\begin{definition} \label{def:tfs2_tfs2}
Let $Q(u,v,w)$ and $R(x,y,z)$ be two $\TFS$ problems.
We say that $Q$ is \emph{counterexample reducible} to $R$ if there 
are polynomial-time functions $f,g,h$ such that for all $u,y,w$ we have
\[
R(f(u), y, h(u,y,w)) \rightarrow Q(u, g(u,y), w).
\]
Clearly this property is transitive, and 
we will write it as $Q \le_c R$. If there are reductions
in both directions then we say $Q$ and $R$ are equivalent, $Q \equiv_c R$.
\end{definition}

The name ``counterexample reducible'' comes from the following interpretation of this definition.
Suppose we are given an input $u$ to $Q$ and want to find a 
solution~$v$. 
We send $f(u)$ as input to an 
Adversary who claims to be able to solve $R$. The Adversary responds that $y$ is a solution to $R$
on this input, that is, they make the coNP assertion that $\forall z \ R(f(u), y, z)$. 
We choose to believe them for now, and make our own
coNP assertion, that $g(u,y)$ 
is a solution to $Q$. 
Now suppose someone comes along and tells us that our coNP assertion
was false, and proves it with a counterexample~$z$ such 
that $\neg R(x, g(u,y), z)$. We can then point to the 
Adversary, and show them a counterexample $h(u,y,w)$
to their coNP assertion about~$R$.

One can extend this to a similar notion of reducibility
between $\mathrm{TF}\Sigma_k$ problems for $k \ge 3$,
or ask about appropriate general notions of reducibility between 
$\mathrm{TF}\Sigma_k$ and $\mathrm{TF}\Sigma_j$.
We will not go down that road here, but note that the
problems $\GI_k$ discussed in the next section can be thought
of as long chains of reductions of the first kind; and
see~\cite{thapen2011higher} for discussion of a particular case of the second question.

We collect together some $\TFS$ and TFNP problems which we will discuss below.

\begin{definition}~\label{def:some_TFS_problems}
\begin{enumerate}
\item
{\sc $\P^\NP$ computation} is the $\TFS$ problem:
given a machine $M$ with an oracle tape for $\NP$,
and a unary time $1^t$,  find a \emph{correct}
computation~$w$ of $M$ running for time $t$. 
Specifically, $w$ lists all queries to the oracle, with YES/NO replies; 
every YES reply must come with a witness that it is correct;
and the $\coNP$ assertion is that there is no witness that any NO reply is false.
\item
{\sc Least number} is the $\TFS$ problem:
given a circuit defining a set $S$ and a string
$x$ in $S$, find the lexicographically least string in~$S$.
That is, output some $y \in S$ with the $\coNP$ assertion
that there is no $z<y$ in~$S$.
\item
{\sc Empty} is the $\TFS$ problem:
given a circuit defining a function~$f$
from a set $[2^n-1]$ of pigeons to a set  
$[2^n]$ of holes, find an empty hole. That is,
find a hole $y$ such that, for all pigeons $x$,
$f(x) \neq y$.
\item
{\sc Weak pigeon} is the $\TFNP$ problem:
given a circuit defining a function $f$ from a set~$[2^{n}]$ of pigeons
to a set~$[2^{n-1}]$ of holes, find distinct pigeons $x, x'$
such that $f(x)=f(x')$.
\item
{\sc Dual weak pigeon} is the $\TFS$ problem:
given a circuit defining a function~$f$
from a set  $[2^n]$ of pigeons to a set  
$[2^{n+1}]$ of holes, find an empty hole.
\item
{\sc Retraction weak pigeon} is the $\TFNP$ problem:
given circuits defining functions $f$ from a set $[2^n]$
of pigeons to a set $[2^{n+1}]$ of holes, and $g$ mapping
in the other direction from holes to pigeons,
find a hole $y$ such that $f(g(y)) \neq y$. 
\end{enumerate} 
\end{definition}

The name {\sc Empty} is
from~\cite{kleinberg2021total}. 
Its weak version {\sc Dual weak pigeon} appears more in the literature,
and in particular was studied intensively, with {\sc Retraction weak pigeon},
in a series of papers formalizing approximate counting in bounded 
arithmetic~\cite{jerabek:dual-wphp, jerabek:apc1, jerabek:apc2}.
It is also sometimes called the \emph{surjective} 
or \emph{onto weak pigeonhole principle}, \emph{range avoidance}, or {\sc 1-Empty}
~\cite{ren2022range, kleinberg2021total}.
Notice that {\sc Retraction weak pigeon} becomes {\sc Left pigeon}
if we replace $[2^{n+1}]$ with $[2^n+1]$
and swap the names of $f$ and $g$.

Our first application of counterexample reduction adapts a result 
of~\cite{buss:axiomatizations}
connecting $\Sigma^b_1$-induction with~$\P^\NP$ computations.

\begin{theorem} \label{the:PNP_LNP_equivalent}
{\sc $\P^\NP$ computation} and {\sc Least Number}
are equivalent under counterexample reductions.
\end{theorem}

\begin{proof}
One direction is simple. Given an instance $(C, x)$ of 
{\sc Least Number}, we define an $\NP$-oracle
machine which finds the least element $y$ of $S$ by binary
search and use this machine as the input to {\sc $\P^\NP$ computation}.
Any counterexample $z$ that $y$ is not least will contradict
one of the NO replies made by the machine, just as in the
proof of Proposition~\ref{pro:QBF1}.

For the other direction, let $(M,1^t)$ be our instance
of {\sc $\P^\NP$ computation}. 
We will show a reduction to the ``greatest number principle".
Consider the following kind of object, which we call
a \emph{precomputation}. 
We simulate $M$ for time $t$;
each time we encounter an NP oracle query, we can either write down YES, in which case we must also write a witness; or
we write NO (which may be incorrect); then we
carry on, using the bit we wrote as the oracle reply. 
We write all the YES/NO bits at the start of the string, and
may assume there exactly $t$ of them, and write any witnesses
at the end. Whether a string is a precomputation is decidable in polynomial time, 
and the set of precomputations is nonempty, since it contains
the string consisting simply of $t$ many NOs.

Now suppose an adversary gives us a string 
$w$ and asserts that it is the lexicographically
greatest precomputation. The only way that $w$ can
fail to be a correct solution to {\sc $\P^\NP$ computation}
is if there is a counterexample in the form of a witness $z$
to an NP query to which $w$ answered NO. Given such 
a~$z$, we construct
a new computation~$w'$ which is the same as~$w$ up to this
query, then flips this NO in~$w$ to~YES, giving $z$
as a witness; then answers every later query simply with~NO.
Then~$w'$ is a precomputation which is lexicographically greater
than $w$, so contradicts the adversary's assertion.
\end{proof}

Any TFNP problem $Q(u,v)$ can also be seen as a $\TFS$ problem in
which the counterexample variable $w$ is not used. 
Between two such problems, counterexample reducibility 
collapses to many-one reducibility.
We also get from this a natural definition of counterexample reducibility of a TFNP
problem to a $\TFS$ problem.\footnote{
See~\cite{kol2022approx} or Proposition~\ref{pro:APPROX_PLS} below for a similar, 
but more complicated, notion of ``$\PLS$ counterexample reducibility''. }

\begin{definition} \label{def:counterexample_TFNP}
A $\TFNP$ problem $Q(u,v)$ is \emph{counterexample reducible}
to a $\TFS$ problem $R(x,y,z)$, written $Q \le_c R$,
if there are polynomial time functions $f,g,h$ such that for all $u,y$ we have
\[
R(f(u), y, h(u,y)) \rightarrow Q(u, g(u,y)).
\]
\end{definition}

We show now that the set of TFNP classes counterexample reducible
to a $\TFS$ problem has a natural
complete problem (under usual many-one reducibility),
which is something like a projection of the $\TFS$ class
onto TFNP, and we will see that some well-known TFNP
problems arise in this way.
The next definition and lemma are
 this section's analogues of Definition~\ref{def:Gamma_TFNP} and 
Proposition~\ref{pro:solvable}.

\begin{definition}
Let $R(x,y,z)$ be a $\TFS$ problem. 
We define {\sc Checkable} $R$ to be the following $\TFNP$ problem:
given a pair $(x,C)$ of an input to $R$ and a circuit $C$,
find $y$ such that $R(x,y, C(y))$ holds.
\end{definition}

One could also call this a \emph{Herbrandized} 
version of the $\TFS$ problem; see for example~\cite{hanika2004herbrandizing}
or the discussion of the Herbrandized ordering principle in~\cite{bkt:fragments}.

\begin{lemma} \label{lem:counterexample_iff_checkable}
Let $R(x,y,z)$ be a $\TFS$ problem and let $Q(u,v)$ be a $\TFNP$ problem.
Then $Q \le_c R$ if and only if $Q \le {\textsc{Checkable}~} R$.
\end{lemma}

\begin{proof}
First suppose $Q \le_c R$ with reduction functions $f,g,h$.
Given an input $u$ to $Q$, let $C$ be a circuit
computing the function $y \mapsto h(u,y)$ on inputs~$y$
of appropriate size (given by the bound on sizes of the solutions
to $R$ on input $f(u)$). We send $(f(u), C)$ as input to ${\textsc{Checkable}~}R$.
If $y$ is a solution to this, we have $R(f(u), y, C(y))$,
hence $R(f(u), y, h(u,y))$, hence by the reduction $Q(u, g(u,y))$.

The other direction is similar. Suppose the reduction 
$Q \le {\textsc{Checkable}~} R$ holds, that is,
we have functions $f_1, f_2, g$, where $f_2(u)$ is a circuit which we write as $C_u$, such that
$R(f_1(u), y, C_u(y)) \rightarrow Q(u, g(y))$. Then the triple $f_1,g,h$ 
is a counterexample reduction $Q \le_c R$,
where $h :  u,y \mapsto C_u(y)$.
\end{proof}


\begin{lemma} \label{lem:TFS_reduction_to_checkable}
Let $Q$ and $R$ be $\TFS$ problems with 
$Q \le_c R$. Then $\chk Q \le \chk R$.
\end{lemma}

\begin{proof}
Given an input $(u,C)$ to {\sc Checkable} $Q$, construct
the input $(f(u), D)$ to {\sc Checkable} $R$, where
$D$ is the circuit mapping $y \mapsto h(u,y,C(g(u,y)))$, and work through the definitions. 
\end{proof}

Item 1 in the next theorem is implicit in~\cite{buss:axiomatizations, buss1994application}.
A form of item 3 is widely used in the bounded arithmetic literature 
for the low-complexity consequences of the
dual weak pigeonhole principle (see e.g.~\cite{thapen2002model, jerabek:apc2}).

\begin{theorem}~\nopagebreak \label{the:TFS_equivalences}
\begin{enumerate}\nopagebreak
\item\nopagebreak
$\chk${\sc $\P^\NP$ computation} and 
$\chk${\sc least number} 
are both equivalent to 
{\sc Localopt}, and so are complete for $\PLS$.
\item
$\chk${\sc empty} is 
equivalent to {\sc Left pigeon}, and so is
complete for $\PPADS$.
\item
$\chk${\sc dual weak pigeon}
is equivalent to {\sc Retraction weak pigeon}.
\end{enumerate}
\end{theorem}

\begin{proof}
For 1, $\chk${\sc least number} takes as input
a set $S$ (as a circuit), an element $x \in S$ and
a circuit $C$. A solution is $y \in S$ such that 
$C(y)$ is not a counterexample to the claim that $y$
is a minimal element of $S$; in other words, such that either 
$C(y) \notin S$ or $C(y) \ge y$. Thus the problem is 
just {\sc Localopt}.
The claim for $\chk${\sc $\P^\NP$ computation}
follows by Theorem~\ref{the:PNP_LNP_equivalent}
and Lemma~\ref{lem:TFS_reduction_to_checkable}.

For 2, $\chk${\sc empty} takes as input a function~$f  : [2^n-1] \rightarrow [2^n]$ (as a circuit), and a circuit~$C$.
A solution is $y \in [2^n]$ such that, if $C(y) \in [2^n-1]$, then $f(C(y)) \neq y$.
If we call $[2^n]$ the set of pigeons and $[2^n-1]$ the set of holes (that is, the opposite 
of how they are named in {\sc empty}), rename $f$ to $g$ and
rename $C$ to $f$, we see that this is {\sc Left pigeon}. 

Item 3 is similar.
\end{proof}

\begin{corollary} \label{cor:empty_separation}
Relative to some oracle,
there is no counterexample reduction in
either direction between {\sc $\P^\NP$ computation}
and {\sc Empty}.
\end{corollary}

\begin{proof}
In the relativized setting we add tapes for an oracle
$B$ to every Turing machine and allow gates for $B$ in every circuit;
in particular, the $\NP$ queries in {\sc $\P^\NP$ computation},
and the circuit used in the definition of the checkable version of a problem,
can query~$B$.
Then everything in this section 
relativizes, so the corollary follows from known
relativized separations between PLS and 
PPADS~\cite{morioka2001classification, goos2022separations}
using Lemma~\ref{lem:TFS_reduction_to_checkable}.
\end{proof}

We showed in Proposition~\ref{pro:QBF1} that {\sc Localopt} $\le \loc\Gqbf^1$. 
We can now show that the converse direction also holds.

\begin{proposition} \label{pro:QBF1_from_PNP}
$\loc\Gqbf^1 \equiv$ {\sc Localopt} and thus $\loc\Gqbf^1$ 
is complete for $\PLS$.
\end{proposition}

\begin{proof}
By Proposition~\ref{pro:QBF1} it is enough to show that $\loc\Gqbf^1 \le$ {\sc Localopt}.
We will show this by proving that $\loc\Gqbf^1$ is counterexample reducible
to $\P^\NP$ {\sc computation}, and appealing to Lemma~\ref{lem:counterexample_iff_checkable}
and item~1 of Theorem~\ref{the:TFS_equivalences}.

As our instance of $\loc\Gqbf^1$ we are given a protocol $(M,1^t)$
in which the Inspector makes queries to a
$\Gqbf^1$ adversary,
and the task is to find a run of this protocol on which 
the Inspector does not find any failure of the $\Gqbf^1$ axioms.
This naturally gives us an instance of 
$\P^\NP${\sc computation} in which we turn each query
``$A(F)$=?'' made by the protocol into corresponding NP query.
Recall that a solution $w$ to this instance is 
a computation of $M$, in which all YES oracle
replies come with a witness; and a counterexample to $w$
is any $z$ which witnesses that one of the NO replies was false.
We have three cases.

In the first case, the run of the protocol 
contained in $w$ does not witness any failure
 of $\Gqbf$. In this case we have a solution
to our initial instance of $\loc\Gqbf^1$.
Otherwise, the run contains some failure of $\Gqbf$, 
and we need to use this to find a counterexample.

There are two different kinds of failure.
In the second case, $w$ asserts that $A(\forall x F)=\top$
and, say, that $A(F_{\restrict x=0})=\bot$, for $F \in \Pi^q_1$
(or, dually, that $A(\exists xF) = \bot$ and $A(F_{\restrict x=0})=\top$,
for $F \in \Sigma^q_1$).
Two things must have happened in $w$:
we made the NP query ``is there an assignment to 
all quantifiers in $\forall x F$ which falsifies it"
and got the answer NO, that is, that $A(\forall x F) = \top$; 
and we made the NP query ``is there an assignment to 
all quantifiers in $F_{\restrict x=0}$ which falsifies it"
and got the answer YES, and recorded a witness to this in~$w$.
From this witness we can construct the required counterexample to 
the NO reply.

In the third case, the run contains a failure
in which we asserted $A(\exists x F) = \top$,
where $\exists x F$ is $\Sigma^q_1$,
and also both $A(F_{\restrict x=0})=\bot$
and $A(F_{\restrict x=1})=\bot$ (or the dual of these).
As in the previous case, these correspond respectively
to a YES reply to an NP query in $w$, with a witness, and two NO replies.
Using the witness we can construct a counterexample
to one of the NO replies.
\end{proof}

Finally we mention that, by 
starting with $\chk${\sc $\P^\NP$ computation} and
incorporating the checking function 
into the computation, we get the following  characterization of PLS,
similar to the Inspector-Adversary framework in Section~\ref{sec:SO_setup}:

\begin{proposition}
The following $\TFNP$ problem {\sc Locally-consistent $\P^\NP$ computation}
is complete for $\PLS$: given a $\P^\NP$ machine $M$, find a computation~$w$
of~$M$ such that the output of $w$  is  not a counterexample
to any NO reply recorded in~$w$.
\end{proposition}

\section{Ramsey and approximate counting} \label{sec:approx_counting}

As discussed in the introduction, 
in Section~\ref{sec:approx}
we give a simplified definition of the approximate-counting TFNP class APPROX
from~\cite{kol2022approx}, using the machinery of counterexample reductions.
In Sections~\ref{sec:long_choice} and~\ref{sec:short_choice}
we give a direct reduction of {\sc Ramsey} to it, and compare it with weakened
versions of the counting problems {\sc Long choice} and {\sc Short choice}
from~\cite{pasarkar2023extremal}.
In Section~\ref{sec:direct_rWPHP_to_PH} we give an illustrative
example of how a very weak form of approximate counting
(namely {\sc Retraction weak pigeon}) lies inside the
TFNP subclasses arising from the polynomial hierarchy.

\subsection{A new definition of APPROX} \label{sec:approx}

The theory $\APC_2$ of~\cite{jerabek:apc2}
can be defined as the theory $T^1_2$ together with a
sentence expressing that the following $\TFS$ problem is total:

\begin{definition}
{\sc $\P^\NP$ retraction weak pigeon},
or {\sc RWPP$_2$}, is the $\TFS$ problem:
given descriptions of $\P^\NP$ machines for functions $F$ and $G$,
where $F$ maps pigeons  $[2^{n}]$
to holes $[2^{n+1}]$ and $G$ maps holes
back to pigeons, find a hole $y \in [2^{n+1}]$ and a 
computation~$w$
showing $F(G(y)) \neq y$.
\end{definition}

Precisely, a solution $(y,w)$ records a run $w$ of the combined $\P^\NP$ machine
that first runs $G$ on $y$, then runs~$F$ on the output
of $G$, and is such that in~$w$ the final output of $F$ is different from $y$.
It records it in the sense of the problem $\P^\NP$ {\sc computation}
in Definition~\ref{def:some_TFS_problems}, listing
witnesses to YES replies and asserting that 
there is no counterexample to any NO reply.

$\APPROX$ is defined in~\cite{kol2022approx} as the class
of $\TFNP$ problems $Q(u,v)$ which are \emph{$\PLS$ counterexample 
reducible} to $\RWPP_2$. This
notion of reducibility is similar to counterexample reducibility
in Definition~\ref{def:counterexample_TFNP} above, 
with the difference
that the function~$g$ taking a solution $(y,w)$ of
$\RWPP_2$ 
to a solution $v$ of $Q$ is
 computed by solving a PLS problem, rather than just by a polynomial-time algorithm.

\begin{theorem} [\cite{kol2022approx}]
$\APPROX$ contains precisely the $\TFNP$ search problems
which are provably total in~$\APC_2$, closed under
many-one reducibility.
\end{theorem}

We simplify the definition by showing that counterexample reducibility by
itself is enough.

\begin{proposition} \label{pro:APPROX_PLS}
A $\TFNP$ problem is in  $\APPROX$ if and only 
if it is counterexample reducible to $\RWPP_2$.
Thus $\chk\RWPP_2$ is complete for $\APPROX$.
\end{proposition}

\begin{proof}
Clearly counterexample reducibility implies PLS
counterexample reducibility, so it is enough to show the other
direction.
Suppose $Q(u,v)$ PLS counterexample reduces to an
instance $(F,G)$ of
$\RWPP_2$, so there is a PLS problem
$P((y,w), s)$ such that, for any solution $(y,w)$
of this instance of $\RWPP_2$ and any solution~$s$ of $P$ on input~$(y,w)$,
in polynomial time we can construct from~$s$
either a solution $v$ of $Q$ or a counterexample~$z$
to $(y,w)$. Let $G'$ be a $\P^\NP$ machine that does the following,
on input~$y$.
First it simulates the computation of $F(G(y))$, 
getting some output $y'$ and recording the computation
as a string $w$; then it uses binary search to solve
the PLS problem $P$ on input $(y,w)$, getting some solution~$s$;
then it outputs simply $y'$. 
We claim that $Q(u,v)$  counterexample reduces to the
instance $(F,G')$ of $\RWPP_2$, which is the original
instance with $G$ replaced by $G'$.

Let $(y,w')$ be a solution of this new instance. We
must construct in polynomial time either a solution $v$
of $Q$ or a counterexample to some NO reply in $w'$.
We recover~$w$ and $s$ (as described in the definition of $G'$ above) from $w'$. 
It must be that~$s$ is the solution to $P$ on input $(y,w)$, since
this is a polynomial-time property that does not involve
any NP query. Therefore we can compute
from $s$ either a solution of $Q$ or a counterexample
to a NO reply in $w$; but such a NO reply is also present
in $w'$, since $w$ is a record of a simulation carried out in~$w'$.
\end{proof}

We now observe that APPROX ``works'' as a class, in that it
contains the counting-related problems which we would expect.

\begin{definition} ~ 
\begin{enumerate} 
\item
{\sc Ramsey} is the $\TFNP$ problem:
given a circuit defining a graph on~$[2^n]$, find 
either a clique or independent set of size
at least~$n/2$.
\item
{\sc Tournament} is the $\TFS$ problem:
given a circuit defining a tournament on~$[2^n]$, find a dominating set
of size at most~$n$.
\item
{\sc Min} is the $\TFS$ problem:
given a circuit defining a binary relation on~$[2^n]$, find either a witness that it
is not a total order, or  a least element.
\end{enumerate}
\end{definition}

The name {\sc Min} is from~\cite{chiari1998witnessing}.
The problem {\sc Checkable min} is similar to the TFNP
problem {\sc Herbrandized ordering principle} (HOP) studied
in~\cite{bkt:fragments}.

\begin{proposition}[\cite{pudlak1990ramsey,jerabek:apc2, bkt:fragments}] \label{pro:approx_list}
{\sc Ramsey}, {\sc Checkable tournament}, {\sc Checkable Min},
{\sc Localopt} and {\sc Weak pigeon} are in $\APPROX$.
\end{proposition}

\begin{proof}
For {\sc Ramsey} we prove this below by an explicit reduction via the problem 
{\sc Weak long choice} in Lemma~\ref{lem:ramsey_to_wlc} and Proposition~\ref{pro:wlc_to_rwpp2},
adapting an argument from~\cite{pudlak1990ramsey}. 
For {\sc Checkable tournament} we only know an indirect proof through bounded
arithmetic: it is shown in~\cite{jerabek:apc2} that the standard counting proof
of the tournament principle is formalizable in the theory $\APC_2$ 
(similarly for the Ramsey theorem), and the reduction 
follows by Proposition~\ref{pro:APPROX_PLS} above.

For {\sc Checkable Min} we sketch a counterexample reduction from 
{\sc Min} to {\sc Tournament}. Given a purported total order $R(x,y)$,
define a tournament by: $x$ beats~$y$ if and only if~$R(x,y)$. Let $S$
be a purported small dominating set. Let $t$ be the $R$-minimal
element of $S$, which we can find by brute force. If $R(z,t)$ for any $z$
(that is, if $z$ is not an $R$-least element), then either $z$ is not dominated
by any element of~$S$, or we have found a violation of a transitivity
or antisymmetry condition for~$R$. (See~\cite{bkt:fragments},
which also contains an argument  for partial
orderings, due to Je\v{r}\'{a}bek.)

Since {\sc Localopt} is equivalent to {\sc Checkable least number principle},
for {\sc Localopt} it is enough to observe that {\sc Least number principle}
is counterexample reducible to {\sc Min}, which is straightforward.

For {\sc Weak pigeon}, given a function $f:[2^n] \rightarrow [2^{n-1}]$
in which we want to find a collision, we set $G$ in $\RWPP_2$ to be just $f$,
and  set $F$ to be the natural $\P^\NP$ machine which searches for an inverse of $f$.
\end{proof}

The next two propositions indicate some limits on what is in APPROX.

\begin{proposition}\label{pro:approx_in_Gqbf3}
{\sc Checkable $\RWPP_2$}, and thus every problem in $\APPROX$,
lies inside $\loc\Gqbf^3$. 
\end{proposition}

\begin{proof}
It may be possible to prove this by a direct reduction similar to that in
Proposition~\ref{pro:rwpp_to_qbf2}, but we give an indirect argument through
bounded arithmetic. Namely,
the $\P^\NP$ weak pigeonhole principle, and thus everything provable in~$\APC_2$,
is provable in the bounded arithmetic theory $T^3_2$ (see~\cite{jerabek:apc2}).
The proposition follows by Corollary~\ref{cor:BA_hierarchy} below.
\end{proof}

\begin{proposition} [\cite{kol2022approx}] \label{pro:CPLS_not_in_approx}
In the relativized setting, $\loc\Gqbf^2$ is not in $\APPROX$. 
\end{proposition}

\begin{proof}
It is shown in~\cite{kol2022approx} that the problem $\CPLS$ is not in
APPROX.  We observe in Section~\ref{sec:T_k_2} below
 that $\loc\Gqbf^2 \equiv \CPLS$.
\end{proof}

Finally we include here a technical
``amplification'' construction for the weak pigeonhole principle,
which we will use in the rest of this section.
In bounded arithmetic this was introduced in~\cite{pww:primes}
and was used in a weaker form, which works in weaker theories, in~\cite{thapen2002model}.
It is similar to the constructions used to amplify hash functions and pseudorandom generators.

\begin{lemma} \label{lem:amplification}
{\sc Retraction weak pigeon} is equivalent to ``weaker'' versions
in which there is a larger gap between the number of pigeons and the number of 
holes. In particular, it is equivalent to the following $\TFNP$ problem:
given functions $F:[2^n] \rightarrow [2^{2n}]$ and 
$G:[2^{2n}] \rightarrow [2^n]$, find $y \in [2^{2n}]$ such that
$F(G(y)) \neq y$. 
\end{lemma}

\begin{proof}
We prove the harder direction.
Suppose we are given an instance $f,g$ of {\sc Retraction weak pigeon},
with pigeons $[2^n]$ and holes $[2^{n+1}]$. 
For $i<n$ let $f_i$ be the function $[2^{n+i}] \rightarrow [2^{n+i+1}]$
given by applying $f$ to the most significant $n$ bits of the input and 
leaving the remaining bits unchanged, and define $g_i$ similarly.
Define $F:x \mapsto f_{n-1}( \dots (f_0(x)))$
and $G:y \mapsto g_0( \dots (g_{n-1}(y)))$.
Suppose $F(G(y)) \neq y$.
Setting $y_n = y$ and $z_0 = G(y)$,
let $y_i = g_i( \dots (g_{n-1}(y)))$
and $z_i = f_{i-1}( \dots (f_0(z_i)))$, tracing the
path of $y$ backwards and forwards through $G$ and $F$.
Then there must be some $i$ such that $y_i = z_i$
but $y_{i+1} \neq z_{i+1}$, from which we can
recover a solution to our instance of {\sc Retraction weak pigeon}.
\end{proof}

\subsection{Long choice} \label{sec:long_choice}

\begin{definition} \hspace{-5pt}\footnote{
The equivalent definition in \cite{pasarkar2023extremal} is of a \emph{subcertificate},
which is a sequence of distinct strings $a_0, \dots, a_k$, with $k \le m+1$, 
such that $P_i(a_0, \dots, a_i, a_j) = P_i(a_0, \dots, a_i, a_{j'})$
for  all distinct $j, j' >i$ (for $j,j' \le k$ and $i \le m{-}1$).
Thus $a_0, \dots, a_k$ is a subcertificate if and only if, in our notation,
$k \le m+1$ and $\vec a$ is a $b_0, \dots b_{k-1}$-sequence for~$\vec P$,
where each $b_i := P_i(a_0, \dots, a_i, a_k)$.
We have changed the notation because we will need to use the values $\vec b$.
}
 \label{def:b-sequence} 
Let  $P_0, \dots, P_{m-1}$ be $0/1$-valued functions where each $P_i$ is of the form
$P_i(x_0, \dots, x_i, z)$ with $i+2$ arguments.
Let $b_0, \dots, b_{\ell-1} \in \{0,1\}$ with $\ell \le m$.
We say that
a sequence of $k$ distinct strings $a_0, \dots, a_{k-1}$ 
is a \emph{$\vec b$-sequence for $\vec P$}  if we have
\mbox{$P_i(a_0, \dots, a_i, a_j) = b_i$} whenever~$i < j$
(over all $i<\ell$ and $j<k$).
\end{definition}

\begin{lemma}[\cite{pasarkar2023extremal}] \label{lem:long_choice_total}
Let $P_0, \dots, P_{n-2}$ be $0/1$-valued functions, where each
$P_i$ takes $i+2$ arguments from $[2^n]$.
Then there exist sequences $b_0, \dots, b_{n-2} \in \{0,1\}$
and $a_0, \dots, a_n \in [2^n]$ such that $\vec a$ is a $\vec b$-sequence for $\vec P$. 
\end{lemma}

\begin{proof} 
Using induction on~$j = 0, \dots, n{-}1$ we construct sequences $B_j := b_0, \dots, b_{j-1}$ and
 $A_j := a_0, \dots, a_{j-1}$, 
where:
all elements of $A_j$ are distinct;
$A_j$ is a $B_j$ sequence for $\vec P$;
and 
$|S(A_j, B_j)| \ge 2^{n-j}$,
where
\[
S(A_j, B_j) :=
\big\{ z \in [2^n] \setminus \{ a_0, \dots, a_{j-1} \} : 
	\forall i \! < \! j \ P_{i}(a_0, \dots, a_i, z) = b_i \big\}
	\]
is the set of strings $z$ such that $A_j^\frown z$ is a $B_j$-sequence
for $\vec P$, where $^\frown$ indicates sequence concatenation.
	
This is true for $j{=}0$ with the empty sequences. Suppose we
have it for~$j {<} n{-}1$. Then $S(A_j, B_j)$ is nonempty, and we choose 
$a_j \in S(A_j, B_j)$ arbitrarily.
Then $S(A_j, B_j)$ decomposes into three disjoint sets
$\{a_j\}$,
$\{ z \in S(A_j, B_j) \setminus \{ a_j \} : P_j(A_j^\frown a_j, z) = 0 \}$
and 
$\{ z \in S(A_j, B_j) \setminus \{ a_j \} : P_j(A_j^\frown a_j, z) = 1\}$.
These last two sets are precisely $S(A_j^\frown a_j, B_j^\frown  0)$ and
$S(A_j^\frown  a_j, B_j^\frown  1)$. One of these
must have size at least $2^{n-j-1}$ and we set $b_j$ to $0$ or $1$ accordingly.

Thus we obtain sequences $B_{n-1}{=} b_0, \dots, b_{n-2}$ and 
$A_{n-1} {=} a_0, \dots, a_{n-2}$ such that $|S(A_{n-1}, B_{n-1})| \ge 2$.
We set the last two elements $a_{n-1}$ and $a_{n}$ to be any two members of $S(A_{n-1}, B_{n-1})$.
\end{proof}

We now define {\sc Long choice} and our weakened version of it.
It follows from the previous lemma that these two problems are total. 

\begin{definition} ~
\begin{enumerate}
\item
{\sc Long choice}~\cite{pasarkar2023extremal} is the $\TFNP$ problem:
given circuits for functions
$P_0, \dots, P_{n-2}$, where $P_i$ takes $i{+}2$ arguments in $[2^n]$,
find sequences $b_0, \dots, b_{n-2} \in \{0,1\}$
and $a_0, \dots, a_n \in [2^n]$ such that $\vec a$ is a $\vec b$-sequence for~$\vec P$. 
\item
{\sc Weak long choice} is the same as  {\sc Long choice}, except that
$[2^n]$ is replaced everywhere with~$[2^{n+1}]$.
\end{enumerate}
\end{definition}

\begin{lemma} \label{lem:ramsey_to_wlc}
{\sc Ramsey} is reducible to {\sc Weak long choice}.
\end{lemma}

\begin{proof}
We adapt the standard proof of the finite Ramsey theorem, just
as in~\cite{pasarkar2023extremal}.
We are given a colouring of all unordered pairs
of nodes in $[2^{n+1}]$ with colours from $\{0,1\}$. 
Our goal is to find a set of  at least $\lceil (n{+}1)/2 \rceil$ nodes 
 all edges between which have the same colour.
We form an instance of {\sc Weak long choice}  by, for each $i<n{-}1$, defining
$P_i(x_0, \dots, x_i, z)$ to be the colour of $\{x_i, z\}$.
Let $b_0, \dots, b_{n-2}$ and $a_0, \dots, a_{n}$ be a solution,
which means that for each $i  < n{-}1$ all edges 
$\{a_i, a_{i+1} \}, \ldots, \{a_i, a_n\}$ have the same colour~$b_i$.
Either~$0$ or~$1$ must appear at least $\lceil (n-1)/2\rceil$ times in
the sequence $b_0, \dots, b_{n-2}$;
suppose for example that $b_{i_0}, \dots, b_{i_{k-1}}$ are all~$0$
where $k = \lceil (n{-}1)/2\rceil$. Then 
$\{ a_{i_0}, \dots, a_{i_{k-1}} \} \cup \{ a_{n} \}$ 
is a monochromatic set of size~$\lceil (n{+}1)/2 \rceil$.
\end{proof}

\begin{lemma} \label{lem:wpp_wlc}
{\sc Weak pigeon} is reducible to {\sc Weak long choice}.
\end{lemma}

\begin{proof}
It is easy to see the {\sc Weak pigeon} is equivalent to a
slightly amplified version in which there are four times as many
pigeons as holes, by a simple version of the argument in Lemma~\ref{lem:amplification}.
So suppose we are given an instance of that problem,
in which we have a function $f:[2^{n+1}] \rightarrow [2^{n-1}]$
and want to find a collision.
We use the same construction as the reduction from 
{\sc Weak pigeon} to {\sc Unary long choice}
in~\cite{pasarkar2023extremal}, 
setting $P_i(x_0, \dots, x_i, z)$ to be the $i$th bit of~$f(z)$
for each $i < n$. Then if $\vec b$, $\vec a$ is a solution
to this instance of {\sc Weak long choice}, we must have
$f(a_{n-1}) = f(a_n) = b$, where $b \in [2^{n-1}]$ has bits 
$b_0, \dots, b_{n-2}$.
\end{proof}

Finally in this subsection we show that {\sc Weak long choice} is in APPROX, that
is, that it is counterexample-reducible to $\RWPP_2$.
For this we will need to construct suitable $\P^{\NP}$ functions from~$[2^n]$ pigeons to $[2^{n+1}]$
holes and vice versa, to give as inputs to $\RWPP_2$.
To define these we adapt the proof of the Ramsey theorem using the weak 
pigeonhole principle in bounded arithmetic, from~\cite{pudlak1990ramsey}. 

We have functions $P_0, \dots, P_{n-2}$ over~$[2^{n+1}]$
as in Definition~\ref{def:b-sequence}.
Suppose we are given
bits $b_0, \dots, b_{m-1}$, for $m < n$, and want to construct a ``canonical''
$\vec b$-sequence $a_0, \dots, a_m$ for $\vec P$, of length~$m{+}1$, if it exists.
We have the following algorithm. First set $a_0 = 0$. Once
we have found $a_0, \dots, a_{j-1}$, set~$a_j$
to be the least $z>a_{j-1}$ 
such that $P_i(a_0, \dots, a_i, z) = b_i$ for each~$i<j$.
Either we eventually find a length-$(m{+}1)$ sequence $a_0, \dots, a_m$ in this
way, in which case we output $a_m$ (which we will show below
contains enough information to recover the computation), or at some point there
is no such~$z$, in which case we stop and output~$0$ as a kind of error code.
Let $F$ be the $\P^\NP$ machine that carries out this algorithm.
$F$ takes input $\vec b$ of length at most~$n{-}1$ and outputs either $0$ or $a_m \in [2^{n+1}]$ as described.

We will say that a
sequence $a_0, \dots, a_m$ is \emph{canonical} if each~$a_j$
is the minimal $z>a_{j-1}$ possible over $a_0, \dots, a_{j-1}$, just as in the algorithm above.

We can also describe an algorithm that goes in the other direction.
It takes as input $y \in [2^{n+1}]$ and tries to find 
sequences $b_0, \dots, b_{m-1}$ and $a_0, \dots, a_m$, for some $m < n$,
where $\vec a$ is a canonical $\vec b$-sequence for $\vec P$ ending
with $a_m = y$. We use essentially the same process as above, with a small change:
once we have found $a_0, \dots, a_{j-1}$, we set $b_{j-1}$ to be 
$P_{j-1}(a_0, \dots, a_{j-1}, y)$, and then proceed as before.
That is, we set~$a_j$ to be the least $z>a_{j-1}$ such that 
$P_i(a_0, \dots, a_i,z) = b_i$ for each~$i<j$. As long as~$a_{j-1} < y$
such a $z$ exists with $z \le y$, since~$y$ itself satisfies the condition
by our choice of $b_0, \dots, b_{j-1}$. 
If we get~$a_j = y$ in this way, with $j < n$, we stop;
we have found the desired sequences $b_0, \dots, b_{j-1}$
and $a_0, \dots, a_j$, and we output $\vec b$. 
Otherwise, if we have run for $n-1$ steps 
and constructed $b_0, \dots, b_{n-2}$ and $a_0, \dots, a_{n-1}$
with~$a_{n-1}$ that is still strictly less than~$y$,
we also stop and output $\vec b$.
Let $G$ be the $\P^\NP$ machine that carries out this algorithm.
$G$ takes input $y \in [2^{n+1}]$ and outputs $\vec b$ 
of length~$n{-}1$ or less. 

\begin{proposition} \label{pro:wlc_to_rwpp2}
{\sc Weak long choice} is counterexample reducible to $\RWPP_2$
and thus is in $\APPROX$.
\end{proposition}

\begin{proof}
We are given functions
$P_0, \dots, P_{n-2}$ 
and want to find bits $b_0, \dots, b_{n-2}$ and 
strings $a_0, \dots, a_n$ in $[2^{n+1}]$ such 
that $\vec a$ is a
$\vec b$-sequence for $\vec P$.

We take the set of pigeons to be the set of binary strings $\vec b$ of length $n{-}1$
or less, which is bijective in polynomial time with the interval $[2^n{-}1]$.
We take the set of holes to be the interval $[2^{n+1}]$.
Our input to $\RWPP_2$ is the pair $(F,G)$, where $F$ maps pigeons
to holes and $G$ maps holes to pigeons.
We get back a solution $(y,w)$,
where $y$ is a hole and~$w$ is a computation of $F(G(y)) \neq y$, coming
with the assertion that no NO reply in~$w$ has a counterexample.
There are now two possibilities. 

The first is that, during the computation described in $w$, 
$G$ found sequences $b_0, \dots, b_{n-2}$ and 
$a_0, \dots, a_{n-1}$ with $a_{n-1} {<} y$
(this is the last case in the description of~$G$ above).
Then, by construction, $a_0, \dots, a_{n-1}, y$ is a 
$\vec b$-sequence for $\vec P$ of the correct length~$n$.
Thus we have a solution to 
our instance of {\sc Weak long choice} and we are done,
as this is a polynomial-time checkable property which does not involve 
any of the $\coNP$ assertions about the strings $a_j$ being minimal.

The other possibility is that $G$ found sequences $b_0, \dots, b_{m-1}$
and $a_0, \dots, a_m$ with $m < n$ and $a_m = y$.
Then the rest of $w$ describes $F$ running on input~$\vec b$.
In this computation  $F$
tries to reconstruct the canonical sequence $a_0, \dots, a_m$ from $\vec b$, and
in particular performs precisely the same searches for minimal strings~$z$
that $G$ did, making the same $\NP$ queries.
Therefore either $F$ finds precisely the same $a_0, \dots, a_m$ and outputs~$y$, which is impossible
since the output of $F$ in $w$ is different from $y$;
or it gets a different YES/NO answer than $G$ did to one of these $\NP$ queries, 
in which case we have a counterexample to~$w$ in the form of a YES witness to some
NP query to which $w$ recorded the answer NO.

We have cheated slightly by using $2^n{-}1$ pigeons, where $\RWPP_2$
strictly uses~$2^n$. However nothing changes if
we extend $F$ to map the missing pigeon to~$0$.
\end{proof}

\subsection{Short choice} \label{sec:short_choice}

\begin{definition} ~
\begin{enumerate}
\item
{\sc Short choice}~\cite{pasarkar2023extremal} is the $\TFS$ problem:
given circuits for functions
$P_0, \dots, P_{n-2}$, where $P_i$ takes $i{+}2$ arguments in $[2^n {-} 2]$,
find sequences $b_0, \dots, b_{k} \in \{0,1\}$
and $a_0, \dots, a_k \in [2^n{-}2]$, with $k \le n{-}2$, such that 
\begin{enumerate}
\item
$a_0, \dots, a_k$ is a $b_0, \dots, b_{k-1}$-sequence for~$\vec P$, and
\item
there is no $a_{k+1} \in [2^n{-}2]$  such that
$a_0, \dots, a_k,  a_{k+1}$ is a  $b_0, \dots, b_k$-sequence for~$\vec P$.
\end{enumerate}
\item
{\sc Weak short choice} is the same as {\sc Short choice},
except that $[2^n{-}2]$ is replaced everywhere with $[2^{n-1}{-}2]$.
\end{enumerate}
\end{definition}

Note that we could equivalently define short choice by keeping the size
of the universe as $[2^n{-}2]$ but replacing $n$ with $n+1$ everywhere
else, so that the instance is given by functions $P_0, \dots, P_{n-1}$
and the bound on the solutions is $k \le n{-}1$. 
Defining it that way shows there is an easy counterexample reduction
to {\sc Short choice}: given an instance
of {\sc Weak short choice} in this form, we can get an instance of 
{\sc Short choice} by ignoring the last function  $P_{n-1}$.

The next lemma justifies that these problems are in $\TFS$.

\begin{lemma}
{\sc Short choice} is total.
\end{lemma}

\begin{proof}
As in Lemma~\ref{lem:long_choice_total} we define
\[S(A_j, B_j) :=
\big\{ z \in [2^n{-}2] \setminus \{ a_0, \dots, a_{j-1} \} : 
	\forall i \! < \! j \ P_{i}(a_0, \dots, a_i, z) = b_i \big\}.
\]
We will inductively build sequences $A_j := a_0, \dots, a_{j-1}$
and $B_j := b_0, \dots, b_{j-1}$,
where: all elements of $A_j$ are distinct;
$A_j$ is a $B_j$ sequence for $\vec P$;
and  $|S(A_j, B_j)| \le 2^{n-j}{-}2$.
We will stop as soon as $S(A_j, B_j)$ becomes empty,
which must happen at $j=n{-}1$ or before. We take 
this $A_j$ and $B_j$ as our desired sequences.

The properties hold for the empty sequences $A_0$ and $B_0$.
Suppose they hold at stage $j<n{-}1$ and that $S(A_j, B_j)$ is nonempty.
Choose $a_j \in S(A_j, B_j)$. Then, as in the proof of
Lemma~\ref{lem:long_choice_total}, we can decompose $S(A_j,B_j)$
into $\{a_j\}$, $S(A_j^\frown a_j, B_j^\frown  0)$ and
$S(A_j^\frown  a_j, B_j^\frown  1)$. By the decomposition, one
of these last two sets must have size at most $2^{n-j-1}{-}2$.
We set $b_j$ accordingly.
\end{proof}

\begin{lemma}
{\sc Retraction weak pigeon} is counterexample reducible to {\sc Weak short choice}.
\end{lemma}

\begin{proof}
By the proof of Lemma~\ref{lem:amplification}, it is
enough to show this for an amplified
version of {\sc Retraction weak pigeon}, namely the problem:
given circuits for functions 
$f:[2^{n-2}] \rightarrow [2^{2n}]$ and 
$g:[2^{2n}] \rightarrow [2^{n-2}]$,
find $y \in [2^{2n}]$ such that $f(g(y))  \neq y$.
Here $[2^{n-2}]$ are the pigeons and $[2^{2n}]$ are the holes.

We define  $P_0, \dots, P_{n-2}$ by taking
$P_i(x_0, \dots, x_i, z)$ to be  the $i$th bit of $f(z)$,
like in Lemma~\ref{lem:wpp_wlc}.
Strictly each $P_i$ should take arguments in $[2^{n-1}{-}2]$;
we  map any $z \ge 2^{n-2}$ to $0$.
Suppose we are given a claimed solution
$b_0, \dots, b_k$ and $a_0, \dots, a_k$
to {\sc Weak short choice}, with $k \le n-2$.
Let $S \subseteq [2^{2n}]$
be the set of holes whose first $k+1$ bits are
 $b_0, \dots, b_k$. Then $|S|=2^{2n-k-1} \ge 2^{n+1}$,
 so by brute-force search we can find $y \in S$ in polynomial time
 such that $y \notin \{ f(a_0), \dots, f(a_k) \}$
 (assuming $n$ is sufficiently large).
 
Let $x=g(y)$. If $f(x) \neq y$ then $y$ is a solution to our amplified instance of {\sc Retraction weak pigeon}
 and we are done.
Otherwise $f(x)=y$,
so $x \notin \{ a_0, \dots, a_k \}$.
Thus $a_0, \dots, a_k, x$ is a  $\vec b$-sequence for 
$\vec P$, so $x$ is a counterexample to the claimed solution
 to {\sc Weak short choice}.
\end{proof}

\begin{proposition} \label{pro:weak_short_approx}
{\sc Weak short choice} is counterexample reducible to $\RWPP_2$.
Therefore any $\TFNP$ problem counterexample reducible to {\sc Weak short choice}
is in $\APPROX$.
\end{proposition}

\begin{proof}
We are given circuits for functions $P_0, \dots, P_{n-2}$. 
For the reduction we will use the same $\P^\NP$ machines $F$ and $G$
as we used in the proof of Proposition~\ref{pro:wlc_to_rwpp2},
with the understanding that strings used as arguments for the functions
$P_i$ now range over $[2^{n-1}{-}2]$ rather than over $[2^{n+1}]$.
We will also swap the role of $F$ and $G$.

Precisely, we now take the set of pigeons to be the interval $[2^{n-1}{-}2]$
and the set of holes to be the set of binary strings of length $n{-}1$ or less,
which we identify with the interval $[2^n{-}1]$. Our input
to $\RWPP_2$ is now the pair $(G, F)$, where $G$ maps pigeons to holes
and $F$ maps holes to pigeons. We get back a solution $(\vec b, w)$,
where $\vec b$ is a hole $b_0, \dots, b_{m-1}$ with $m<n$ and
$w$ is a computation of $G(F(\vec b)) \neq \vec b$ 
coming with the assertion that no NO reply in $w$ has a counterexample. 
There are two possibilities.

The first is that $F$ output the error code~$0$. This means that in~$w$, 
$F$ found a canonical $b_0, \dots, b_{k}$-sequence $a_0, \dots, a_{k}$ for some $k < m$,
but could not find any $z>a_{k}$
with $P_i(a_0, \dots, a_i, z) = b_i$ for all $i \le k$. 
We take $b_0, \dots, b_{k}$ and $a_0, \dots, a_{k}$ as our solution to 
our instance of {\sc Weak short choice}. This satisfies the first
part of the definition of {\sc Weak short choice}, because a fortiori
$a_0, \dots, a_{k}$ is a $b_0, \dots, b_{k-1}$-sequence for $\vec P$.
Suppose that we are given a counterexample $a_{k+1}$ for the second part
of the definition, that is, $a_0, \dots, a_{k+1}$ is a $b_0, \dots, b_k$-sequence
for~$\vec P$. If $a_{k+1} > a_k$, then $a_{k+1}$ is a counterexample
to the NO reply in $w$ asserting that there is no $z>a_k$ with this property.
Otherwise $a_{k+1}$ lies inbetween $a_i$ and $a_{i+1}$ for some $i<k$,
and witnesses that $a_{i+1}$ was not the least possible choice
to extend $a_0, \dots, a_i$, so necessarily is a counterexample to some
NO reply made in $w$ during the search for $a_{i+1}$.

The second possibility is that $F$ finds $a_0, \dots, a_m$ which is 
a $b_0, \dots, b_{m-1}$-sequence for $\vec P$, and outputs $y = a_m$.
Then, as in the proof of Proposition~\ref{pro:wlc_to_rwpp2},
$G$ must fail to reconstruct $a_0, \dots, a_m$ and $b_0, \dots, b_{m-1}$
from~$y$, since we are given that it does not output~$\vec b$. 
The only way this can happen is that $G$ and $F$ at some point
received different replies to the same $\coNP$ query, which gives
us a counterexample to~$w$.

This is not quite an instance of $\RWPP_2$, since we have
$2^{n-1}{-}2$ rather than $2^n$ pigeons and $2^n{-}1$ rather than $2^n$ holes.
We extend $G$ to map the two extra pigeons to the extra hole,
and $F$ to map the extra hole to one of the extra pigeons.
\end{proof}

\subsection{A direct reduction of {\sc Retraction weak pigeon} to {\sc Local-}$\Gqbf^2$}
\label{sec:direct_rWPHP_to_PH} 

Finally in this section we give another example of the search problems that 
live inside the $\loc\Gqbf^k$ hierarchy above the level of PLS, 
by showing that {\sc Retraction weak pigeon} is reducible to $\loc\Gqbf^2$.
The existence of such a reduction follows already from  Corollary~\ref{cor:BA_hierarchy}
and known proofs of the retraction weak pigeonhole principle in the theory $T^2_2$,
for example based on the proof in~\cite{maciel2000new},
as would a similar proposition for {\sc Weak pigeon}. 
However we include a self-contained proof because it is interesting to see what
an explicit ``combinatorial'' reduction, not going through logic, looks like. 
A less-direct
reduction is sketched in~\cite{krajivcek2007np}. 

\begin{proposition} \label{pro:rwpp_to_qbf2}
{\sc Retraction weak pigeon} is reducible to $\loc\Gqbf^2$.
\end{proposition}

\begin{proof}
We adapt the proof of the weak pigeonhole principle in the bounded arithmetic
theory $S^3_2$, based on~\cite{pww:primes}.
By Lemma~\ref{lem:amplification} 
we may assume that we only need to solve an instance of the amplified version
of {\sc Retraction weak pigeon},
that is, we have
$f:[2^{n}] \rightarrow [2^{2n}]$ and 
$g:[2^{2n}] \rightarrow [2^{n}]$
and want to find $y \in [2^{2n}]$ such that $f(g(y))  \neq y$.
We think of $f$ as pair of functions $f_0, f_1$ outputting two $n$-bit strings,
and of $g$ as a two-argument function $g(y_0, y_1)$ that combines two $n$-bit strings into one.
We first outline the bounded arithmetic proof, which is by a diagonalization argument.
Given a binary string $w = w_0, \dots, w_{m-1}$, 
write $F_w$ for the function $x \mapsto f_{w_{m-1}}(\dots (f_{w_0}(x)))$.

Consider the full binary tree of depth $n$, where we  label the root (at the top)
with some string $x$, and for every node with label $z$ we label its left and right
children with $f_0(z)$ and $f_1(z)$. Then $F_w(x)$ is the label of the node at address $w$.
We let $H$ be a diagonal function $[2^n] \rightarrow [2^n]$, 
with $H(w)=1$ if $F_w(w) = 0$ and $H(w)=0$ otherwise. 
By construction, there is no label $x$ for the root of the tree that will result in every leaf $w$
getting the label $H(w)$.

To complete the proof, we will try to use the function $g$ to construct exactly such a label~$x$.
In bounded arithmetic this is set up as a length-$n$ backwards $\Sigma^b_3$ induction on $i$,
with inductive hypothesis:
for every node $w$ at level $i$ in the tree there is
a label $z$ for $w$ such that for every path $v$ from $w$ to a leaf $wv$,
we have $F_v(z)=H(wv)$.
This is true at the leaves (level~$n$)
 and we can use $g$ for the inductive step, 
constructing labels for level $i$ from labels for level $i+1$.
The only way for the inductive hypothesis to fail is if we encounter
a solution to {\sc Retraction weak pigeon}. Otherwise
we conclude the hypothesis for level~$0$, the root, which gives us 
a label~$x$ contradicting our diagonal construction.

This argument can be turned directly into a protocol for $\Gqbf^3$ that would solve the problem.
However, by results in~\cite{buss:axiomatizations} if a TFNP
problem is provably total in~$S^3_2$, that is, using an induction of the syntactic form described, 
then there is already a $\P^\NP$ machine any computation of which will 
provably solve the problem; such a machine is
essentially the same thing as a protocol for $\Gqbf^2$. We describe such a protocol explicitly.

Consider a branch $b$ in the tree. 
Define a \emph{good labelling} of $b$ as follows:
as we traverse $b$ from the root to a leaf, each time we go to a right-hand child
 we have a label for the left-hand child $w$ which gives rise, using $F$, to a labelling
of the subtree below $w$ in which every leaf $wv$ gets the label $H(wv)$, where $H$
is as above.
Being a good labelling is a $\Pi^q_1$ property. The empty labelling is
a good labelling for the leftmost branch, and any good labelling for the rightmost
branch leads to a contradiction, since we can use it to construct in polynomial
time using $g$ a label $x$ for the root such that for every leaf $v$ we have $F_v(x)=H(v)$
(or if the construction fails, it gives us a witness to {\sc Retraction weak pigeon}, as below).
So we can use binary search over branches $b$, asking $\Sigma^q_2$ queries,
to find $b$ such that the Adversary asserts that $b$ has a good labelling and $b+1$ does not.
We then try to construct a good labelling for $b+1$ by using $g$ to combine
labels from $b$; either we will force the Adversary into a contradiction, 
or we will find string $z_0, z_1, z$ such that $z=g(z_0, z_1)$ but 
$(f_0(z), f_1(z)) \neq (z_0, z_1)$, which is a witness  to {\sc Retraction weak pigeon}.
\end{proof}

\section{Decision-tree classes and propositional proofs} \label{sec:dt}

We show that some previously-studied TFNP
problems are characterized by the constant-depth Frege
and (unrestricted) Frege proof systems, in the sense of~\cite{goos2022separations}.
We will show in the next section that these are equivalent to our problems
$\loc\Gqbf^k$ and $\loc\Gqbf$.
We will need to work with a slightly different, non-uniform
version of relativized TFNP, which works well with propositional proof complexity.

\subsection{Decision-tree TFNP and characterizations}

The following definitions are from~\cite{goos2022separations},
following~\cite{hubacek2024tfnp}.

\begin{definition}
An (abstract) \emph{total search problem} is a total relation
$R \subseteq \{0,1\}^{r} \times \Out$,
where~$r \in \NN$ and $\Out$ is some finite set.
We think of strings $x \in \{0,1\}^r$ as \emph{inputs}
and elements of $\Out$ as possible \emph{solutions}.
For a family $(R_a)_{a \in \NN}$ of such problems we will
use the notation $R_a, r_a, \Out_a$.

A family $(R_a)_{a \in \NN}$ is in $\TFNP\dt$ if $r_a$ is at most 
quasipolynomial in~$a$ and, for each $y \in \Out_a$,
there is a decision tree $T_{a,y}$ of depth $\poly(\log a)$,
querying~$x$ and deciding whether $(x,y) \in R_a$.
\end{definition}

\begin{definition} \label{def:dt_reduction}
Let $R \subseteq \{0,1\}^r \times \Out_R$
and $S \subseteq \{0,1\}^s \times \Out_S$ be two total search problems.
A \emph{reduction} of $R$ to $S$ is a pair of functions
$f:\{0,1\}^r \rightarrow \{0,1\}^s$ and 
$g:\{0,1\}^r \times \Out_S \rightarrow \Out_R$
satisfying, for all $x \in \{0,1\}^r$ and all $y' \in \Out_S$,
\[
(f(x), y') \in S \longrightarrow (x, g(x,y')) \in R.
\]
The reduction has \emph{depth $d$} if each bit of $f(x)$
and all functions $x \mapsto g(x,y')$, for each $y' \in \Out_S$,
are computable by depth-$d$ decision trees, querying bits of~$x$.
\end{definition}

We will often  call something a search problem,
or a CNF, when strictly we should say it is a family of such things. It should
be clear from the context what is meant.

\begin{definition}
Let $R$ and $S$ be total search problem families in $\TFNP\dt$. 
We say $R$ \emph{is reducible to $S$}, or $R \le S$, if
for each $a$ there is $q$ quasipolynomial in $a$ such that 
there is a reduction of $R_a$ to $S_q$ of depth $\poly (\log a)$.
If also $S \le R$ we say that $R$ and $S$ are equivalent.
\end{definition}

We will call a CNF family $F=(F_a)_{a \in \NN}$ \emph{narrow}
if each $F_a$ has quasipolynomially in $a$ many clauses and variables,
and width at most $\poly(\log a)$.

\begin{definition}
For an unsatisfiable CNF $F$, $\Search(F)$ is the total search problem:
given a total assignment to the variables of $F$, find a false clause in~$F$.

For a family $F$ of unsatisfiable CNFs, $\Search(F)$ is the corresponding
family of search problems. If $F$ is a narrow family, then $\Search(F)$ is 
in $\TFNP\dt$.
\end{definition}

Now let $P$ be a propositional proof system and let $\mu_P$ be a measure of
proof size, appropriate for $P$;
for example $\mu_P(\pi)$ might be the number of symbols in $\pi$,
or, if $P$ is an algebraic system, the degree of~$\pi$.
For a CNF $F$ we write
$\mu_P(F) := \min \{ \mu_P(\pi) : \pi$ is a $P$-refutation of $F \}$.

\begin{definition} \label{def:characterization}
Let $R \in \TFNP\dt$. We say that the class of search problems reducible to $R$
is \emph{characterized by} proof system $P$ under measure $\mu_P$ if the following 
holds:
for every unsatisfiable narrow CNF family $F$, 
\[
\Search(F) \le R \ \longleftrightarrow \ \mu_P(F_a) \le \poly(\log a).
\]
We may also say that $R$ is characterized by $P$ under $\mu_P$, with the same meaning.
\end{definition}

The following characterizations (and more) are collected
together in~\cite{goos2022separations}.
These hold 
for a larger range of parameters, but we focus on the most important
case described in Definition~\ref{def:characterization} above, 
where the bound on the depth and on $\mu_P$ is of the form
$\poly(\log a)$.
The complexity classes below are the decision-tree versions of
the standard $\TFNP$ classes; see~\cite{goos2022separations}
for full definitions.

\begin{theorem}
We have the following pairs of $\TFNP\dt$ classes and proof
systems the characterize them, where the measure
 $\mu_P(\pi)$ is $\log(\mathrm{size}(\pi)) + \mathrm{degree}(\pi)$
 for the algebraic systems
 and $\log(\mathrm{size}(\pi)) + \mathrm{width}(\pi)$ for the clausal systems.
\begin{enumerate}
\item
$\FP\dt$ and treelike resolution \cite{lovasz95}
\item
$\PLS\dt$ and resolution \cite{bkt:fragments}
\item
$\PPA\dt$ and Nullstellensatz over $\mathbb{F}_2$ \cite{goos2019adventures}
\item
$\PPAD\dt$ and unary Nullstellensatz over $\NN$ \cite{goos2022separations}
\item
$\PPADS\dt$ and unary Sherali-Adams \cite{goos2022separations}
\end{enumerate}
\end{theorem}

We can add two more characterizations,
proved in Corollaries~\ref{cor:LK_k_2} and~\ref{cor:Gqbf_characterize} below.

\begin{proposition} \label{pro:proof_characterizations}
Under 
the measure $\mu_P(\pi) = \log(\mathrm{size}(\pi))$,
\begin{enumerate}
\item
$\loc\Gqbf^{k+2}$ 
is characterized by $\LK_{k+\hf}$ for $k \in \NN$ \\
(equivalently $\loc\Gqbf^{k+1}$ is characterized by depth-$(k{+}\hf)$ Frege)
\item
$\loc\Gqbf$
is characterized by the Frege system.
\end{enumerate}
\end{proposition}

Here we expect that item 1 could also be expressed as:
$\loc\Gqbf^k$ is characterized by depth-$k$ Frege, under
a suitably-defined measure of the form $\log(\mathrm{size}) +
($fan-in at depth 1$)$.

These are not really new results, but come from translating similar
characterizations in bounded arithmetic into the propositional 
setting of $\TFNP\dt$.
Firstly, \cite{skelley2011provably} defines a hierarchy
of TFNP problems $\GI_k$ and shows that they capture
the TFNP problems provably total in the theory $T^k_2$.
We show that this means that, unsurprisingly,
$\GI_{k+2}$ is characterized by the system $\LK_{k+\hf}$
known to correspond to $T^{k+2}_2$, and we show later
that $\GI_k \equiv \loc\Gqbf^{k}$. 
Secondly, \cite{beckmann2017np} defines a TFNP problem
$\FCon$ and shows that it captures the TFNP problems
provably total in the theory~$U^1_2$. We show
this means that $\FCon$ is characterized by the system Frege
known to correspond to $U^1_2$, and show later that
$\loc\Gqbf \equiv \FCon$.

We will work with fragments of the sequent calculus system LK,
in which every line is a propositional sequent.
However we will treat it as a \emph{cedent calculus} or \emph{Tait calculus},
in which lines are \emph{cedents}, that is, sets of formulas written as comma-separated sequences,
which behave semantically as disjunctions. 
This is a cosmetic change -- see e.g.~\cite{buss1998introduction}.
If we do not put any limit on the logical depth of formulas in a proof,
this is equivalent to the Frege system, in which lines are
propositional formulas. If we limit formulas to depth $k \in \NN$
then depth-$k$ LK, or $\LK_k$, arguably corresponds to depth-$(k{+}1)$ Frege,
since a line in such an LK proof is a cedent of depth-$k$ formulas,
and so the naturally corresponding line in a Frege proof 
would be the disjunction of those formulas, with depth $k{+}1$.
 For example,
resolution can be seen as depth-$1$ Frege (as we can model a clause is a 
depth-$1$ disjunction) or as $\LK_0$ (as we can model a clause
as a cedent consisting of depth-$0$ formulas, that is, literals).
An LK refutation of a CNF $F$ is an
LK derivation of the empty sequent in which we can use
each clause of $F$ (considered as a cedent) as an axiom.

We define $\LK_{k+\hf}$ to be an extension of $\LK_k$
to allow small additional fan-in gates, instead of just literals, at the bottom level.
Precisely, an $\LK_{k+\hf}$ proof is an $\LK_{k+1}$ proof
in which all subformulas at depth 1 (that is, conjunctions or disjunctions
of literals) have size at most log of the size of the proof -- see e.g~\cite{krajicek2019proof}.
In particular, for proofs of size quasipolynomial in some size parameter~$a$,
the bottom fan-in will be $\poly(\log a)$.
So $\LK_{\hf}$ is $R(\log)$~\cite{krajicek2001weak} 
and we could take $\LK_{-\hf}$ to be
polylogarithmic-width resolution.

\subsection{Extension by shallow decision trees}

We show some technical results that will help us connect
$\LK_{k+\hf}$ to reasoning with shallow decision trees. 
These are more-or-less implicit in the way $\LK_{k+\hf}$ is 
standardly used in switching lemmas and in translations of bounded arithmetic.

If $b$ is a branch in a decision tree, then each edge in $b$
is labelled with the literal assigned to true along that edge. 
We naturally identify $b$ with the conjunction of these literals,
and  write $\neg b$ for the clause that is the negation of this conjunction.

\begin{definition}
Let $T$ be a decision tree querying variables from some set $\vec x$.
The \emph{extension axiom} for $T$ is a CNF defining a new variable 
$e_T$, whose value should depend on the value of $\vec x$ based
on the decision made by the tree. The CNF consists of the clauses 
$\neg b \vee e_T$ for each accepting branch $b$ of $T$,
and $\neg b \vee \neg e_T$ for each rejecting branch $b$ of $T$.
\end{definition}

\begin{definition}
Let $P$ be a propositional proof system. We define a new system called
\emph{$P$ extended by depth-$d$ decision trees}. Given a CNF $F$,
a refutation of $F$ in this system is a $P$-refutation of $F \wedge A$,
where $A$ is a collection of extension axioms defining new variables
by decision trees, of depth at most~$d$, over the original variables of $F$.
We define \emph{$P$ extended by shallow decision trees} to be this system,
where the maximum depth $d$ is 
the logarithm of the size of the refutation.
\end{definition}

\begin{lemma} \label{lem:decision_tree_to_half_LK}
For $k \in \NN$, the system $\LK_{k + \hf}$ is quasipolynomially equivalent
to the system $\LK_k$ extended by shallow decision trees.
\end{lemma}

\begin{proof}
We show this first for $\LK_{0 + \hf}$, that is, Res(log).

Let $\Pi$ be an $\LK_{0 + \hf}$ refutation of a CNF $F$. 
For each conjunction $B = y_1 \wedge \dots \wedge y_s$ appearing 
in $\Pi$ we introduce an extension variable $e_B$ representing~$B$,
using the decision tree that queries $y_1, \dots, y_s$ in order,
rejecting as soon as any~$y_i$ is false and accepting if they are all true. 
We turn $\Pi$ into a new object $\Pi'$ by replacing each conjuction with its 
extension variable; we do not touch individual literals which are not conjunctions,
and we may assume that disjunctions have already been expanded into comma-separated
sequences of literals, appearing directly in the cedent.
We turn $\Pi'$ into an $\LK_0$ refutation using the extension axioms for the new variables,
 by filling in gaps between the lines
of $\Pi'$. 
To do this, we observe that the extension axioms for $e_B$
consists  of the clauses $\{ \neg y_1, \dots, \neg y_s, e_B \}$
(for the only accepting branch) and 
$\{ y_1, \neg e_B \}$, $\{ \neg y_1, y_2, \neg e_B \}$, \dots ,
$\{ \neg y_1, \dots, \neg y_{s-1}, y_s, \neg e_B \}$ (for the rejecting branches). 
From the rejecting-branch clauses, for each~$i$ we can derive $\{ y_i, \neg e_B \}$ 
by resolution. Hence from any cedent $\Gamma, e_B$ in $\Pi'$ we can derive
$\Gamma, y_i$ for each~$i$. Now suppose $\Pi$ contained a $\wedge$-introduction instance
deriving $\Gamma, B \wedge z$ from $\Gamma, B$ and $\Gamma, z$.
In $\Pi'$, we must derive $\Gamma', e_{B \wedge z}$ from $\Gamma', e_B$ and 
$\Gamma', z$, where $\Gamma'$ is whatever $\Gamma$ turned into in $\Pi'$.
To do this we derive each $\Gamma', y_i$ from $\Gamma', e_B$
and then resolve all of these, together with $\Gamma', z$,
against the accepting-branch clause $\{ \neg y_1, \dots, \neg y_s, \neg z, e_{B \wedge z} \}$
of the axiom for $e_{B \wedge z}$, to derive $\Gamma', e_{B \wedge z}$ as required.
Other rules of $\LK_{0 + \hf}$ are handled similarly.

For the other direction, suppose we have an $\LK_0$ refutation $\Pi$ of $F \wedge A$,
where $A$ is a set of shallow decision tree extension axioms.
For each extension variable $e_T$, in each cedent in $\Pi$
we replace every positive occurrence $e_T$ with 
the set of conjunctions $\{ b : b$ is an accepting branch in $T \}$
and each negative occurrence $\neg e_T$ with 
the set of conjunctions $\{ b : b$ is a rejecting branch in $T \}$.
This gives a new object $\Pi'$, in which we need to fill in some gaps to get
an $\LK_{0 + \hf}$ refutation of $F$.
Consider first a clause in $A$ of the form $\neg b \vee e_T$ for $b$ an accepting
branch of $T$. In $\Pi'$ this will become the cedent
$\neg b \cup \{ c : c$ is an accepting branch in $T \}$,
where $\neg b$ is a set of literals (representing a disjunction).
In particular this will have as a subcedent $\neg b \cup \{ b \}$,
which is derivable in $\LK_{0 + \hf}$ from propositional axioms.
Clauses in $A$ for rejecting branches are handled the same way.
For a resolution step in $\Pi$ on an extension variable $e_T$ 
we must show how to derive $\Gamma'$ from
$\Gamma', b_1, \dots, b_s$
and $\Gamma', c_1, \dots, c_t$ in $\LK_{0+\hf}$, where
$b_1, \dots, b_s$ and $c_1, \dots, c_t$ list respectively
the accepting and rejecting branches of $T$, written as conjunctions.
We can do this by a series of cuts, observing that for each pair
$b_i$ and $c_j$ there is a literal in $b_i$ whose negation is in $c_j$.

For $k>0$, we simulate $\LK_{k+\hf}$ in $\LK_k$ with decision trees
by replacing each depth-($k{+}\hf$) formula with a depth-$k$ formula
built  from extension variables. For the other direction,
we can replace an extension literal in a depth-$k$ formula
with either a narrow CNF or a narrow
DNF, whichever is needed to limit the depth.
\end{proof}

\begin{definition}
For $R$ in $\TFNP\dt$
we can write a CNF expressing that $R_a$ has no solution, namely
$
\CNF(R_a) :=
\bigwedge_{y \in \Out_a} \bigwedge \{ \neg b : b
\textrm{ is an accepting branch
of } T_{a,y} \}.
$
$\CNF(R)$ is a narrow and unsatisfiable family of CNFs.
\end{definition}

Similar lemmas to the following, showing that a
TFNP reduction $R \le S$ implies the existence of 
a propositional proof of the totality
of $R$ from the totality of $S$, appear 
in~\cite{buresh2004relativized, buss2012propositional}.

\begin{lemma} \label{lem:reduction_to_treelike}
Let $F$ be a family of unsatisfiable narrow CNFs
and let $S \in \TFNP\dt$. Suppose that
$\Search(F) \le S$.
Then there are height-$\poly(\log a)$ treelike
resolution derivations
$
F_a \wedge A \vdash \CNF(S_{q}),
$
where $F_a$ is in variables $\vec x$; $q$ is quasipolynomial in~$a$;
$\CNF(S_q)$ is in a distinct set of variables $\vec z$; and
$A$ is a set of decision-tree extension axioms defining each $z_i$
by a  $\poly(\log a)$-depth decision trees $T_i$ querying $\vec x$.
\end{lemma}

\begin{proof}
The reduction $\Search(F) \le S$ is given by functions $f$ and $g$,
where $f$ maps assignments to $F_a$ to inputs to 
$S_q$, and $g$ maps possible solutions $y' \in \Out_{S_q}$ to
indices of clauses in $F_a$.
Suppose $S_q$ takes $s$ bits of input, which we may
write as $z_1, \dots, z_s$; similarly we
write the variables of $F_a$ as $x_1, \dots, x_r$. 
The functions above are 
computed by decision trees querying $\vec x$ variables, 
and in particular, by expanding out the definition of a reduction, for each $z_i$
there is a tree~$T_i$, 
such that, when each $z_i$ is computed by $T_i$,
if $(\vec z, y') \in S_q$ then $g(\vec x, y')$ is a
clause in $F_a$ which is false under~$\vec x$
(where we are abusing notation and using $\vec x$ 
as a name both for the variables and the assignment to them).

By definition $\CNF(S_q)$ contains,
for each $y' \in \Out_{S_q}$,
a clause $\neg b$
for each accepting branch
$b$ of the tree $T_{y'}$, where $T_{y'}$ accepts exactly the inputs $\vec z$
for which $y'$ is a solution. 
We claim that each such $\neg b$ can be derived by a shallow treelike
resolution derivation from $F_a \wedge A$. To show this it suffices
to describe a ``Prover'' protocol
which starts with a partial assignment to $\vec z$
satisfying $b$ and, using $\poly(\log a)$ many
adaptive queries to $\vec x$ variables, 
extends this to a partial assignment falsifying some
clause from~$F_a \wedge A$.

Suppose $b$ contains variables $z_{j_1}, \dots, z_{j_d}$.
The protocol is: first run the trees $T_{j_1}, \dots, T_{j_d}$
(which query $\vec x$ variables). Then either we have falsified
a clause of $A$, or the values of $z_{j_1}, \dots, z_{j_d}$
are computed correctly from $\vec x$ by these trees.
Then run the tree for $g(\vec x, y')$ (also querying $\vec x$ variables),
getting the index for some clause $C$ of $F_a$. Finally 
query all variables in $C$.
Since $b$ was an accepting branch of $T_{y'}$, we have that
$(\vec z, y') \in S_q$
(where $\vec z$ is any extension of $b$
to a total assignment) and thus, from the definition of a reduction,
the clause $C$ named by $g(\vec x, y')$ is false in~$\vec x$.
Thus we have grown our assignment to falsify some clause of $F_a$
as required.
\end{proof}

\subsection{Characterization of $\GI_{k+2}$ by $\LK_{k+\hf}$} \label{sec:GI_characterization}

\begin{definition}[\cite{skelley2011provably}] \label{def:GI_k}
{\sc $k$-turn game induction}, or $\GI_k$, is the following $\TFNP\dt$ problem.
The input consists  of
\begin{enumerate}
\item
For each $i \in [2^n]$ a relation
$G_i(x_1, \dots, x_k)$, taking $k$ arguments in $[2^n]$, where
$G_0$ is everywhere true and $G_{2^n-1}$ is everywhere false
\item
For each $j=1,\dots , k$ and each $i \in [2^n{-}1]$, a
function $f^j_i$ from $[2^n]^j$ to $[2^n]$.
\end{enumerate}
The relations $G_i(x_1, \dots, x_k)$  encode
a sequence of games $G_i$ between two players A and~B,
where $x_j$ is a move of A for odd $j$ and of B for even~$j$.
We interpret $G_i(x_1, \dots, x_k)$ 
as expressing the relation ``B wins the play $x_1, \dots, x_k$ in $G_i$''.

A solution consists of $i \in [2^n{-}1]$ and two sequences
$x_1, \dots, x_k$ and $x'_1, \dots, x'_k$, 
representing plays in respectively
$G_i$ and $G_{i+1}$, such that
\begin{enumerate}
\item
$x_j = f^j_i(x'_1, x_2, x'_3, x_4, \dots , x_{j-1}, x'_j)$
for every odd $j$
\item
$x'_j = f^j_i(x'_1, x_2, x'_3, x_4, \dots, x'_{j-1}, x_j)$
for every even $j$
\item
$G_i(x_1, \dots, x_k)$ is true, that is, B wins this play of $G_i$
\item
$G_{i+1}(x'_1, \dots, x'_k)$ is false, that is, A wins this play of $G_{i+1}$.
\end{enumerate}
\end{definition}

This is slightly simplified from~\cite{skelley2011provably},
as we have removed the explicit winning strategies for $B$ and $A$
in respectively 
in the first and last game, replacing them with the condition that
$G_0$ is everywhere true and $G_{2^n-1}$ is everywhere false;
now the negation of 1. - 4. implies that $f^2_0, f^4_0, \dots$ give a winning strategy
for $B$ in $G_1$ and $f^1_{2^n-2}, f^3_{2^n-2},
\dots$ give a winning strategy for $A$ in $G_{2^n-2}$.

To see that the problem is total observe that,
since B always wins $G_0$ and A always wins~$G_{2^n-1}$,
 by induction there is some $i$
such that B has a winning strategy for $G_i$ and A has a winning strategy for 
$G_{i+1}$. That is, there is some $i$ such that
\begin{align*}
& \forall x_1 \exists x_2 \forall x_3 \dots  \ \ G_i(x_1, \dots, x_k) \quad \textrm{and} \\
& \exists x'_1 \forall x'_2 \exists x'_3 \dots \neg G_{i+1}(x'_1, \dots, x'_k),
\end{align*}
where we omitted the quantifiers at the end of the sequences $\dots$ since they will
depend on whether $k$ odd or even.
Substituting $x_1$ with $f^1_i(x'_1)$, then $x'_2$ with $f^2_i(x'_1, x_2)$, etc.
shows that a solution exists.

\begin{lemma} \label{lem:GI_refutation}
$\CNF(\GI_{k+2})$ has polynomial-size refutations in $\LK_k$.
\end{lemma}

\begin{proof}
This is shown for $\GI_3$ in~\cite{skelley2011provably} 
(which is more concerned with showing quasipolynomial-size refutations
in $\LK_{k+ \hf}$, corresponding to bounded arithmetic).
We sketch the proof for $k{=}2$, which easily generalizes to larger~$k$.
We may take $\CNF(\GI_{4})$ to consist of, for each
tuple $(i, x_1, \dots, x_4, x'_1, \dots, x'_4)$ in $[2^n{-}1] \times [2^n]^8$,
a clause expressing that the tuple does not satisfy all 
the conditions 1. to 4. at the end of Definition~\ref{def:GI_k}.
We write this clause as a cedent
\[
x_1 \neq f^1_i(x'_1), \
x'_2 \neq f^2_i(x_2), \
x_3 \neq f^3_i(x'_3), \
x'_4 \neq f^4_i(x_4), \
\neg G_i(\vec x), \
G_{i+1}(\vec x')
\]
where for clarity we have omitted some arguments from the functions~$f^j_i$ 
and where e.g. $x_1 \neq f^1_i(x'_1)$ is shorthand for a disjunction
(written as a comma-separated list of literals in the cedent)
expressing that one of the bits of the value of $f^1_i(x'_i)$ (which is encoded in binary)
is different from one of the bits of~$x_1$.
By resolving the cedents corresponding to all possible values of $f^4_i(x_4)$
(keeping everything else fixed) and applying $\bigvee$-introduction, we obtain
\[
x_1 \neq f^1_i(x'_1), \
x'_2 \neq f^2_i(x_2), \
x_3 \neq f^3_i(x'_3), \
\neg G_i(\vec x), \
\bigvee\nolimits_{\! x'_4} G_{i+1}(\vec x')
\]
Applying $\bigwedge$ introduction over the cedents for all values of $x_4$, 
we obtain
\[
x_1 \neq f^1_i(x'_1), \
x'_2 \neq f^2_i(x_2), \
x_3 \neq f^3_i(x'_3), \
\bigwedge\nolimits_{x_4} \! \! \! \neg G_i(\vec x), \
\bigvee\nolimits_{\! x'_4} G_{i+1}(\vec x')
\]
We repeat the previous two steps for $f^3_i(x'_3)$ and  $x'_3$ to obtain
\[
x_1 \neq f^1_i(x'_1), \
x'_2 \neq f^2_i(x_2), \
\bigvee\nolimits_{\! x_3} \bigwedge\nolimits_{x_4} \! \! \! \neg G_i(\vec x), \
\bigwedge\nolimits_{x'_3} \bigvee\nolimits_{\! x'_4} G_{i+1}(\vec x').
\]
We then resolve the cedents corresponding to all values of $f^2_i(x_2)$
but do not do $\bigvee$-introduction, to get 
\begin{equation} \label{eq:vees}
x_1 \neq f^1_i(x'_1), \
\bigvee\nolimits_{\! x_3} \bigwedge\nolimits_{x_4}  \! \! \! \neg G_i(\vec x), \
\bigdoublevee\nolimits_{\! x'_2}
\bigwedge\nolimits_{x'_3} \bigvee\nolimits_{\! x'_4} G_{i+1}(\vec x')
\end{equation}
where the notation $\bigdoublevee$ indicates that this is not a single
disjunction but rather stands for a comma-separated list of formulas in the cedent
(which is semantically the same as a disjunction).
We have this cedent for every $i$, $x_1$, $x'_1$ and $x_2$.

Now fix $i$ and suppose inductively that we have derived, for every $x_1$, the
cedent 
\[
M_{i,x_1} := \ \
\bigdoublevee\nolimits_{\! x_2} 
\bigwedge\nolimits_{x_3} \bigvee\nolimits_{\! x_4} G_{i}(\vec x).
\]
We want to derive $M_{i+1, x'_1}$ for every $x'_1$. Fix $x'_1$.
For each $x_1$, we cuts  $M_{i, x_1}$ against~(\ref{eq:vees})
for every value of $x_2$, cutting on the formula 
$\bigvee_{x_3} \bigwedge_{x_4} \! \!  \neg G_i(\vec x)$.
This gives
\[
x_1 \neq f^1_i(x'_1), \
\bigdoublevee\nolimits_{\! x'_2}
\bigwedge\nolimits_{x'_3} \bigvee\nolimits_{\! x'_4} G_{i+1}(\vec x').
\]
Observe that here $x_1$ only appears in the disjunction $x_1 \neq f^1_i(x'_1)$.
So we may resolve all these cedents together on this disjunction,
that is, over all values of $f^1_i(x'_1)$, to get $M_{i+1, x'_1}$ as required.
\end{proof}

\begin{definition}
{\sc 1-Ref$(P)$}, or
{\sc Narrow CNF reflection for $P$}, is the $\TFNP\dt$ problem:
given a narrow CNF $F$, a purported refutation $\Pi$ of $F$ in system $P$,
and an assignment~$X$ to the variables of $F$,
find either
a clause in $F$ that is false under~$X$, or
a syntactical mistake in $\Pi$.
\end{definition}
Here the ``1'' in the name reflects the logical depth
of a narrow CNF.
See~\cite{buss2023tfnp} for some related recent work about this problem.
We observe the following simple connection with
the false-clause search problem.

\begin{lemma} \label{lem:1-ref}
Let $F = (F_a)_{a \in \NN}$ be a family of unsatisfiable narrow CNFs,
which has quasipolynomial size (in $a$) $P$-refutations.
Then $\Search(F) \le ${\sc 1-Ref$(P)$}.
\end{lemma}

\begin{proof}
Let $(\pi_a)_{a \in \NN}$ be  a family of quasipolynomial-sized
$P$-refutations of $(F_a)_{a \in \NN}$. 
The input to $\Search(F_a)$ is an assignment $X$ to the variables
of $F_a$, and a solution is a false clause in $F_a$. We 
need to turn $X$ into an input $f(X)$ for {\sc 1-Ref$(P)$},
where each bit of $f$ is computed by a narrow decision tree. 
We take $f(X)$ to be simply the triple $(F_a, \pi_a, X)$,
where we write $F_a$ and $\pi_a$ to mean these objects coded
appropriately as binary strings. These strings are fixed for
each $a$, so each bit of $f$ is computed by a trivial decision tree. 
Any solution to this instance of {\sc 1-Ref$(P)$} can only 
be a false clause in $F_a$, since by assumption $\pi_a$ does
not contain any syntactical mistakes.
\end{proof}

\begin{theorem} [Essentially \cite{skelley2011provably}] \label{the:GI_characterize}
$\GI_{k+2}$ is characterized by $\LK_{k+\hf}$ under the measure $\log($size$)$,
for $k \in \NN$.
\end{theorem}

\begin{proof}
Fix $k \in \NN$ and let $F$ be a family of unsatisfiable narrow CNFs.
We want to show that $\Search(F) \le \GI_{k+2}$ if and only if $F$
has quasipolynomial size $\LK_{k+\hf}$ refutations.

First suppose that $\Search(F) \le \GI_{k+2}$. By Lemma~\ref{lem:reduction_to_treelike}
for each $a$
there is a height-$\poly(\log a)$ treelike resolution derivation 
$F_a \wedge A \vdash \CNF(\GI_{k+2})$ (where we are suppressing the
size parameter $2^n$ of $\GI_{k+2}$, which is here quasipolynomial in $a$)
for some set $A$ of $\poly(\log a)$-depth decision-tree extension axioms.
By Lemma~\ref{lem:GI_refutation}, $F_a \wedge A$ has a
quasipolynomial size $\LK_k$ refutation.
By Lemma~\ref{lem:decision_tree_to_half_LK} we conclude that $F_a$ has a quasipolynomial
size $\LK_{k+\hf}$ refutation.

Now suppose that $F_a$ has a quasipolynomial
size $\LK_{k+\hf}$ refutation. By Lemma~\ref{lem:decision_tree_to_half_LK}, 
$F_a \wedge A$ has a quasipolynomial size $\LK_k$ refutation $\Pi$, 
for some set $A$ of $\poly(\log a)$-depth decision-tree extension axioms.
We now use the proof-theoretic constructions from~\cite{skelley2011provably}.
We first rewrite $\Pi$ in the version of depth-$k$ propositional
sequent calculus used there under the name $\PK_k$. This does not allow
depth-$k$ disjunctions (only conjunctions) and uses a combination of quantifier elimination
rules and resolution, rather than allowing cuts on complex formulas.
We then use the construction in~\cite[Theorem 21]{skelley2011provably}
to construct from $\Pi$ a refutation of $F_a \wedge A$ in the system
$\PK^0_k$ introduced there, with at most quasipolynomial blowup in size.
This is a kind of ``deep inference'' system, whose rules allow changes to an arbitrary
subformula of a formula in a cedent, and is set up so that each application
of a rule changes a subformula essentially by a single literal.
For a similar system see the \emph{symmetric calculus} of~\cite{pudlak2021canonical}.

It is shown in~\cite[Theorem 4]{skelley2011provably} that
$\oneref(\PK^0_k) \le \GI_{k+2}$. By Lemma~\ref{lem:1-ref},
using our $\PK^0_k$ refutation of $F_a \wedge A$
we have a reduction of $\Search(F_a \wedge A)$ to 
(an instance of) $\oneref(\PK^0_k)$.
It remains to show that there is a reduction $\Search(F_a) \le
\Search(F_a \wedge A)$. Given any assignment $X$ to the variables
of $F_a$, we can extend $X$
to an assignment $X'$ also to the extension variables defined in $A$, and such that $X'$ satisfies $A$,
by evaluating the corresponding decision trees.
Then any false clause in $F_a \wedge A$ under $X'$
must already be a false clause in $F_a$ under $X$. This gives
the required reduction.
\end{proof}

\subsection{Characterization of {\sc FCon} by Frege}

\begin{definition} [\cite{beckmann2017np}]
{\sc Frege consistency}, or $\FCon$, is the following $\TFNP\dt$ problem.
The input consists of a purported Frege proof of a contradiction, that is,
a Frege derivation using only the standard logical axioms and rules, that
ends with the contradictory formula $(z_1 \wedge \neg z_1)$.
This is coded as a sequence of symbols, namely 
variables, parentheses, logical connectives and commas, annotated
with some extra information about the larger-scale structure 
(for example, each parenthesis comes with a pointer to the matching
parenthesis). A solution consists of a syntactic error in
this purported proof.
\end{definition}

This problem is total, because Frege is in fact consistent.
This fact has short Frege proofs:

\begin{lemma} \label{lem:CNF_FCon}
$\CNF(\FCon)$ has quasipolynomial-sized Frege refutations.
\end{lemma}

\begin{proof}
In~\cite{beckmann2017np} it is shown that the totality of $\FCon$,
considered as a type-2 TFNP problem (see Section~\ref{sec:bounded}),
is provable in the second-order bounded arithmetic theory~$U^1_2$.
(It was already known that $U^1_2$ proves that Frege is consistent, 
see e.g.~\cite{krajicek1995frege}, but~\cite{beckmann2017np}
works with this particular formalization in terms of $\FCon$.)
By the translation of $U^1_2$ into propositional logic~\cite[Theorem 9.1.6]{krajicek1995frege},
it follows that there are quasipolynomial-size\footnote{
The translation in~\cite{krajicek1995frege} strictly shows that proofs
in the weaker system $U^1_1$ translate into quasipolynomial-size Frege.
However the difference between $U^1_1$ and $U^1_2$ is an axiom
expressing ``$x^{\log x}$ exists for every $x$''. Incorporating this into
the translation gives at most a quasipolynomial blow-up in proof size, 
and in particular still gives us quasipolynomial-size Frege. }
Frege proofs that $\FCon$ is total.
That is, there are proofs deriving a formula of the form $\bigvee_i B_i$,
where $B_i$ expresses ``there is a mistake in the proof at symbol~$i$''
and where the propositional variables express which symbol occurs
at which location etc. as in $\FCon$. From this we can easily construct
a Frege refutation of $\CNF(\FCon)$.
\end{proof}

The following theorem is not surprising. Its interest lies in the fact that,
by the characterization of $\FCon$ in terms of $U^1_2$ 
in~\cite{beckmann2017np}, we have several other 
TFNP problems equivalent to $\FCon$, including
$\loc\Gqbf$ (see Section~\ref{sec:bounded}).

\begin{theorem}[Essentially~\cite{beckmann2017np}] \label{the:FCon_characterize}
$\FCon$ is characterized by the Frege system under the measure $\log($size$)$.
\end{theorem}

\begin{proof}
Let $F$ be a family of unsatisfiable CNFs. We want to show
that $\Search(F) \le \FCon$ if and only if $F$ has quasipolynomial-size
Frege refutations.

First suppose that $\Search(F) \le \FCon$. 
As in the proof of Theorem~\ref{the:GI_characterize},
by Lemma~\ref{lem:reduction_to_treelike}
and the refutation of $\CNF(\FCon)$ in Lemma~\ref{lem:CNF_FCon}
there are quasipolynomial-sized refutations of $F$ in Frege extended
by shallow decision trees. Hence there are 
quasipolynomial-sized refutations of $F$ in Frege, since we
can simply substitute formulas computing the decision trees for 
the extension variables.

For the other direction, suppose $F$ has quasipolynomial-size Frege
refutations. As in the proof of Theorem~\ref{the:GI_characterize},
we could use this to get a reduction $\Search(F) \le \oneref($Frege$)$.
However, we want to use the consistency search problem $\FCon$ rather
than the reflection principle $\oneref$. As is observed in~\cite{buss2023tfnp},
 it is not hard to move between these. We give a direct reduction
of $\Search(F)$ to $\FCon$.

As in~\cite{beckmann2017np} we will allow ourselves to use symbols
$\top$ and $\bot$ in our formal Frege proofs. These can be taken as
standing for formulas $(z_1 \vee \neg z_1)$ and $(z_1 \wedge \neg z_1)$,
and we may assume that $\top$ has a short Frege derivation. 
Let $\Pi$ be our refutation of~$F_a$.
Suppose we are given an input to $\Search(F_a)$, in the form of 
a total assignment $X$ to the variables $\vec x$ of $F_a$. 
For every variable $x_i$, 
replace every occurrence of~$x_i$ in~$\Pi$
with $\top$ or $\bot$ depending on the value of $x_i$ in $X$,
and call the result $\Pi'$. 
Each clause~$C$ in $F_a$ gives rise to an initial
formula $C'$ in $\Pi'$.
If $C$ is true in $X$ (which can be checked by a shallow decision tree) 
then $C'$ can be obtained
by weakening from~$\top$; we add such a derivation of $C'$ to $\Pi'$.
If $C$ is false in $X$, we do nothing.
Call the new derivation~$\Pi''$.
Now we can consider $\Pi''$ as a Frege derivation of $\bot$
using the standard axioms and rules, but containing some mistakes,
namely the clauses~$C'$ where $C$ is false in~$X$,
which are not derived by any axiom or rule. 
Therefore we can use a call to $\FCon$ to find
such a clause, as required.
\end{proof}

\section{Results using bounded arithmetic} \label{sec:bounded}

We show that the problems $\GI_k$ from the last
section are equivalent to $\loc\Gqbf^k$, and the problem
$\FCon$ is equivalent to $\loc\Gqbf$. Rather than giving
direct reductions, we use characterizations of $\GI_k$ and $\FCon$
from bounded arithmetic. We will not give full definitions
of the bounded arithmetic theories here; see e.g.~\cite{buss1998first, krajicek2019proof}.
We are implicitly using the relativized versions of these theories, where the theories
are equipped with relation symbols to talk about unspecified oracles, since we are 
working with relativized search problems.

\subsection{Type-2 TFNP}

A \emph{type-2} TFNP problem 
(see e.g.~\cite{beame1995relative}) takes as input a size parameter~$a$
and a string $\beta$ of length quasipolynomial in $a$,
which we treat like an oracle. The problem is specified
by a polynomial-time oracle machine $R(x,y;B)$
and a polynomial~$p$. A solution to the input $(a, \beta)$
is any $y \in [2^{p(|a|)}]$ such that $R(a,y; \beta)$. The problem
must be total, that is, some solution must exist for
every input.

A type-2 TFNP problem naturally gives rise to a family
of $\TFNP\dt$ search problems, by converting the polynomial-time oracle
machine into a family of decision trees. On the other hand,
a definition of a $\TFNP\dt$ problem is usually sufficiently uniform
that it can also be read as a definition of a type-2 problem.
This is the case for the problems $\GI_k$ and $\FCon$ in the previous section.

For search problems where the inputs are combinatorial objects
given by circuits, the type-2 version is naturally formed by letting
the objects be given by oracles instead.
In this section we want to work with type-2 versions
of $\loc\Gqbf$, $\loc\Gqbf^k$ and $\loc\Gps$
where the situation is messier, since they talk about QBFs/circuits
and evaluations. For these to work correctly,
we also let these circuits be relativized, allowing
them to include gates for an oracle $D$.
Axiom 3 in $\Gqbf$ was
``if $F$ is quantifier-free then $A(F) = \eval(F)$"
but we now understand it as 
``if $F$ is quantifier-free then $A(F) = \eval(F; D)$''
where $\eval(F; D)$ evaluates a circuit with $D$-gates
by making a call to $D$. We understand $\eval$ similarly 
in $\Gps$.

\begin{definition}
Type-2 $\loc\Gqbf$ takes as input a size parameter~$a$,
an oracle-sized string~$\beta$ coding a protocol $\pi$ in the form
of tree of depth $\log a$ querying $A$,
and another oracle-sized string~$\delta$. A solution
is a branch in the protocol, given by a sequence $\alpha$ of 
queries and replies to~$A$, such that the label of the branch does
not witness the failure of any axiom of $\Gqbf$ on $\alpha$,
where we understand axiom 3. as above and evaluate $D$-gates
in QBFs using the string $\delta$.

Type-2 $\loc\Gps$ is similar -- it takes as input a size parameter~$a$,
a string~$\beta$ encoding the protocol, and a string $\delta$ encoding
an oracle that may be queried at each step by the $\mathrm{PSPACE}$ machine it models.
\end{definition}

The proof of equivalence of $\loc\Gqbf$ and $\loc\Gps$ still goes through in
this setting.
Throughout this section, when we talk about these problems
we mean the type-2 versions defined above.

\subsection{The problems {\sc Local-}$\Gqbf^k$ and the theories $T^k_2$}
\label{sec:T_k_2}

The theories $T^k_2$ 
of Buss~\cite{buss1985bounded} 
can be defined by taking a basic theory for polynomial-time functions 
and adding full induction for every formula from the class $\Sigma^b_k$, which 
expresses $\Sigma^p_k$ properties.
It was shown in~\cite{buss1985bounded} to be closely connected to $\P^{\Sigma^p_k}$
computations, and  from the point of view of which TFNP problems it proves total we can 
take it as having the same strength as the statement ``every $\P^{\Sigma^p_k}$ machine
has a computation''.

\begin{proposition} \label{pro:GI_loc_k}
For $k \in \NN$ the problems $\loc\Gqbf^k$ and $\GI_k$
are equivalent.
\end{proposition}

\begin{proof}
For the direction $\loc\Gqbf^k \le \GI_k$,
recall that $\loc\Gqbf^k$ is the problem: given a protocol for the Inspector
in which she queries QBFs in $\Sigma^q_k \cup \Pi^q_k$
(which may use oracle gates), find a run of the protocol
in which none of the Adversary's replies are contradictory. The theory $T^k_2$ proves
that such a run exists, since we can simulate the protocol, with correct
replies of the Adversary,
by a computation of a $\P^{\Sigma^p_k}$ machine, and such
computations provably exist and are well-behaved in $T^k_2$.
Since $T^k_2$ proves that $\loc\Gqbf^k$ is total,
it follows by the characterization of $\GI_k$ and $T^k_2$ 
in~\cite{skelley2011provably} that $\loc\Gqbf^k \le \GI_k$.

For the other direction, it is enough to show that 
$\GI_k$ is solvable over $\Gqbf^k$. For this we follow
the informal proof of $\GI_k$ sketched before 
Lemma~\ref{lem:GI_refutation}.
Namely, we can write the formula
$\forall x_1 \exists x_2 \forall x_3 \dots G_i(x_1, \dots, x_k)$,
expressing that B has a winning strategy in $G_i$,
as a QBF in $\Pi^q_k$ -- notice that this QBF
needs to make ``oracle'' calls to access the relation~$G_i$,
which was part of the input to our instance of $\GI_k$. 
The Inspector uses binary search with queries of this kind
to find~$i$ such that B has a winning
strategy for $G_i$ and A
has a winning strategy for $G_{i+1}$.
Then she can query the first move $x'_1$ of $A$ in $G_{i+1}$ in $A$'s strategy; apply the 
function $f^1_i$
to get a candidate first move $x_1$ of $A$ in $G_i$; query what $B$'s move $x_2$
in response should be, using $B$'s strategy in $G_i$, and so on. 
This will eventually give plays for $G_i$ and $G_{i+1}$
which match according to the functions $f^1_i, f^2_i, \dots$ 
but in which B wins in $G_i$ but not in $G_{i+1}$,
as required.
\end{proof}

\begin{corollary} \label{cor:LK_k_2}
The decision-tree versions of $\loc\Gqbf^k$ and $\GI_k$ are equivalent.
Therefore by Theorem~\ref{the:GI_characterize},
$\loc\Gqbf^{k+2}$ is characterized by $\LK_{k+\hf}$ under the measure
$\log($size$)$.
\end{corollary}

From Proposition~\ref{pro:GI_loc_k} it follows that other previously-known
TFNP problems are equivalent to $\loc\Gqbf^k$, since they 
are equivalent to $\GI_k$. These include 
a $k$-round version of the linear local improvement principle, mentioned in 
Section~\ref{sec:U12}~\cite{kol2011so};
 finding a pure Nash equilibrium in a succinctly-given $k$-turn, zero sum 
game, where the players are only allowed polynomial-time improvements to 
their strategies~\cite{pudlak2012alternating}; and
see also~\cite{beckmann2009polynomial, beckmann2010characterising}.
For the case~$k{=}2$, $\loc\Gqbf^2$ is equivalent to the problem CPLS~\cite{krajivcek2007np}.

The first part of the following is really a corollary of the proof of Proposition~\ref{pro:GI_loc_k},
and the second part of the properties of $\GI_k$:

\begin{corollary} \label{cor:BA_hierarchy}
$T^k_2$ proves that $\loc\Gqbf^k$ is total.
Conversely, any  $\TFNP$ problem provably total in $T^k_2$ is reducible to $\loc\Gqbf^k$.
\end{corollary}

\begin{corollary} \label{cor:hierarchy_separate}
In the decision-tree setting, for $k {\ge} 1$ the problem 
$\loc\Gqbf^k$ does not contain, and is not contained in,
any of  $\PPP$, $\PPAD$, $\PPADS$ or $\PPA$.
\end{corollary}

\begin{proof}
In this setting, suppose $\loc\Gqbf^k$ contains any of these classes. Then it contains PPAD,
and so in particular {\sc Onto pigeon} is reducible to $\loc\Gqbf^k$.
It follows by Corollary~\ref{cor:LK_k_2} that the 
propositional bijective pigeonhole principle has quasipolynomial
size constant-depth Frege proofs, which contradicts
known lower bounds~\cite{pitassi1993exponential, krajivcek1995exponential}. 
%
For the other direction, it is enough to observe that 
none of these four classes contains PLS,  by~\cite{buresh2004relativized, goos2022separations}.
\end{proof}

\subsection{The problem {\sc Local-}$\Gqbf$ and the theory $U^1_2$} \label{sec:U12}

 $U^1_2$ was introduced in~\cite{buss1985bounded} 
as a second-order theory of bounded arithmetic, with polynomial-length
induction for formulas that can existentially quantify over oracle-sized objects.
It was shown in~\cite{buss1985bounded} that it is closely connected to PSPACE, and
from our point of view, being only interested in which TFNP problems are provably
total in a theory, it has the same strength as the statement ``every PSPACE machine
has a computation''.

\begin{proposition} \label{pro:U_1_2_Gps}
$U^1_2$ proves that the search problem $\loc\Gps$ is total.
Conversely, any  $\TFNP$ problem provably total in $U^1_2$ is reducible to $\loc\Gps$.
\end{proposition}

\begin{proof}
For  one direction, since PSPACE computations exist provably in $U^1_2$ 
our general proof of totality, Proposition~\ref{pro:loc_totality}, goes through.
For the other direction, suppose a type-2 TFNP problem $R$ is provably total
in~$U^1_2$. We may write this as 
\[
U^1_2 \vdash \forall U \ \forall a \ \exists v {\in} [2^{p(|a|)}] \ R(a, v ; U)
\]
where $p$ is the polynomial bound on the size of solutions.
That is, for all oracle-sized inputs $U$ (which we model in bounded arithmetic
as a second-order variable) and all size parameters $a$, there is a solution~$v$
of $R$ on this input. We then follow the first three paragraphs of the proof
of Theorem~16 of~\cite{kol2011so}. This starts with what we have written
above, in slightly different notation (for example exploiting that second-order
objects record their own length, so we do not need an explicit size parameter).
It shows that we can replace $U^1_2$ with a weak theory $\PV^{+}$
plus the assertion that an oracle $\alpha$ correctly records
the sequence of configurations computed by a PSPACE computation with transition function~$f$.
It constructs functions~$g$ and~$h$, which we may consider as parts of
an Inspector protocol trying to solve $R$ over $\Gps$.
These have the property that for every $z$, if we use $z$ to construct
a sparse oracle $A_z$, then either $g$ finds a failure of some $\Gps$
axioms in $A_z$ through querying $A_z$, or $h$ finds a solution to $R$ by querying $A_z$.
In other words, the protocol successfully solves $R$ over $\Gps$.
\end{proof}


\begin{corollary}
The problems $\loc\Gps$ and $\FCon$ are equivalent.
\end{corollary}

\begin{proof}
This follows from~\cite[Theorem 1.2]{beckmann2017np}, which has the same
statement as Proposition~\ref{pro:U_1_2_Gps} but with $\FCon$
in place of $\loc\Gps$.
\end{proof}

\begin{corollary} \label{cor:Gqbf_characterize}
The decision-tree versions of $\loc\Gqbf$ and $\FCon$ are equivalent.
Therefore by Theorem~\ref{the:FCon_characterize},
$\loc\Gqbf$ is characterized by the Frege system under the measure
$\log($size$)$.
\end{corollary}

Proposition~\ref{pro:U_1_2_Gps} gives us some other,
more purely combinatorial search problems equivalent to $\loc\Gqbf$,
from the literature on TFNP in bounded arithmetic,
namely the \emph{linear local improvement principle}, about a progressive
sequence of labellings of a graph, and the related \emph{one-round rectangular
local improvement principle}~\cite{kol2011so, beckmann2014improved}.

{\small
\bibliographystyle{alphaabbr}
\bibliography{main}}

\end{document}